\documentclass[english, 12pt]{article}
\usepackage[T1]{fontenc}
\usepackage[utf8]{inputenc}
\usepackage{geometry}
\usepackage[parfill]{parskip}
\usepackage{float}
\usepackage{xcolor}
\usepackage{amsthm}
\usepackage{enumerate}
\usepackage{amsmath}
\usepackage{amssymb}
\usepackage{hyperref}
\usepackage{thm-restate}
\usepackage{bbm}

\allowdisplaybreaks
\makeatletter
\floatstyle{ruled}
\newfloat{algorithm}{tbp}{loa}
\providecommand{\algorithmname}{Algorithm}
\floatname{algorithm}{\protect\algorithmname}

  \theoremstyle{plain}
  \newtheorem{lem}{\protect\lemmaname}[section]

  \newtheorem{conjecture}{Conjecture}
  \newtheorem*{conjecture*}{Conjecture}
  \newtheorem{question}{Question}

  \newcounter{resultcounter}
  \newtheorem{corollary}{Corollary}

 \def\thmhead@plain#1#2#3{%
  \thmname{#1}\thmnumber{\@ifnotempty{#1}{ }\@upn{#2}}%
  \thmnote{ {\the\thm@notefont#3}}}
\let\thmhead\thmhead@plain

  \newtheorem{thm}{\protect\theoremname}
  \newtheorem*{thm*}{\protect\theoremname}

\def\@seccntformat#1{\@ifundefined{#1@cntformat}
   {\csname the#1\endcsname\quad} 
   {\csname #1@cntformat\endcsname}
}
\let\oldappendix\appendix
\renewcommand\appendix{
    \oldappendix
    \newcommand{\section@cntformat}{\appendixname~\thesection\quad}
}

\makeatother

\usepackage{babel}
  \providecommand{\factname}{Fact}
  \providecommand{\lemmaname}{Lemma}
\providecommand{\theoremname}{Theorem}

\newcommand{\abs}[1]{\left|{#1}\right|}
\newcommand{\norm}[1]{\left\|{#1}\right\|}
\newcommand{\innerproduct}[2]{\left\langle {#1}, {#2}\right\rangle}
\let\ip=\innerproduct
\DeclareMathOperator{\linspan}{span}

\DeclareMathOperator{\vol}{vol}

\newcommand{\R}{\mathbb{R}}

\newcommand{\0}{\emptyset}

\DeclareMathOperator{\poly}{poly}

\newcommand{\defeq}{\overset{\text{def}}{=}}
\newcommand{\eps}{\varepsilon}
\newcommand{\unif}{\in_{\text{R}}}
\DeclareMathOperator{\erfc}{erfc}
\newcommand{\bone}{\mathbbm{1}}
\newcommand{\ftab}{\phantom{0000}}

\newcommand{\T}{\intercal}

\title{Spherical Discrepancy Minimization and Algorithmic Lower Bounds for Covering the Sphere}
\author{Chris Jones \and Matt McPartlon}
\date{\today}

\begin{document}
\global\long\def\norm#1{\Vert#1\Vert}


\global\long\def\tr#1{\mbox{tr}\left(#1\right)}

\maketitle
\begin{abstract}
    Inspired by the boolean discrepancy problem, we study the following optimization problem which we term \textsc{Spherical Discrepancy}: given $m$ unit vectors $v_1, \dots, v_m$, find another unit vector $x$ that minimizes $\max_i \innerproduct{x}{v_i}$. We show that \textsc{Spherical Discrepancy} is APX-hard and develop a multiplicative weights-based algorithm that achieves optimal worst-case error bounds up to lower order terms. We use our algorithm to give the first non-trivial lower bounds for the problem of covering a hypersphere by hyperspherical caps of uniform volume at least $2^{-o(\sqrt{n})}$. We accomplish this by proving a related covering bound in Gaussian space and showing that in this \textit{large cap regime} the bound transfers to spherical space. Up to a log factor, our lower bounds match known upper bounds in the large cap regime.
\end{abstract}

\section{Introduction}
Let $S^{n-1} = \{x \in \R^n:\norm{x}_2=1\}$ denote the surface of the sphere in $\R^n$. Suppose we have a collection of unit vectors $v_1, v_2, \dots, v_m \in S^{n-1}$. The goal of this work is to study the following optimization problem on the sphere, which we call \textsc{Spherical Discrepancy},
\[ \min_{x \in S^{n-1}} \max_i \innerproduct{v_i}{x}.\]
The name comes from the boolean discrepancy problem in which $x$ is required to be in $\{-1,+1\}^n$. The unit-norm requirement on $x$ is crucial, otherwise the minimum is always either zero (achieved by the zero vector) or unbounded. \textsc{Spherical Discrepancy} is a relaxation of the boolean discrepancy problem and a primary task of this paper is to adapt and improve upon algorithms from the boolean domain.

The \textsc{Spherical Discrepancy} problem is intimately connected to the following covering problem on $S^{n-1}$: given $m$, what is the smallest value $\theta$ such that $m$ spherical caps of angular radius $\theta$ can cover $S^{n-1}$? Given a unit vector $v \in S^{n-1}$, corresponding to a pole, the cap associated with this pole is given by $\left\{ x\in S^{n-1}:\left\langle v,x\right\rangle \geq\cos\theta\right\}$. Thus, a set of caps of angular radius $\theta$ covers the sphere if and only the value of the \textsc{Spherical Discrepancy} instance on the poles is at least $\cos\theta$. This connection allows for a natural translation from algorithms for \textsc{Spherical Discrepancy} to algorithmic lower bounds for the cap covering problem: given a sparse set of caps, by running the \textsc{Spherical Discrepancy} algorithm on the poles we can produce a witness that lies outside of all the caps. 

In this paper we develop an algorithm for \textsc{Spherical Discrepancy} and use the above connection to prove the first non-trivial lower bounds for this sphere covering problem in what we call the \textit{large cap regime}. In this regime, the volume of each spherical cap is required to be significantly large relative to the volume of the sphere. More precisely, each cap must cover a $2^{-\sqrt{n}}$ fraction of $S^{n-1}$. Outside of the large cap regime (and the regime in which the caps are very tiny), to the best of the authors' knowledge, no nontrivial lower bounds are known.

\subsection{Prior Work on \textsc{Spherical Discrepancy}}
There are a few immediate algorithmic observations we can make about the \textsc{Spherical Discrepancy} problem. 
If $S^{n-1}$ is replaced by a convex body, \textsc{Spherical Discrepancy} can be efficiently solved via standard convex programming techniques. In the case of the non-convex sphere, there are two simple exact algorithms with worst-case runtimes $m^{\Omega(n)}$: set up \textsc{Spherical Discrepancy} directly as a quadratic program, or compute a spherical Voronoi diagram then search for a cell of maximal radius. Studying the problem under a different name, Petkovi{\'c} et al~\cite{sphereCoveringQuadraticProgramming} develop a sophisticated recursive algorithm which is efficient in practice though it still has worst-case runtime $m^{\Omega(n)}$. There are additional algorithms and applications in $\R^3$, see for example Cazals and Loriot~\cite{3dCoveringDataStructure}. For a number of applications with large $n$, see \cite{sphereCoveringQuadraticProgramming}. 

It is NP-hard to construct a PTAS for \textsc{Spherical Discrepancy} i.e. to output a solution whose value is within a factor $1+\eps$ of the true optimum (cf. Section~\ref{sec:hardness-proof}). As in the boolean discrepancy problem, we are interested in $\poly(m, n)$-time approximation algorithms which achieve a worst-case bound independent of the true optimum value of the instance. With this goal in mind, it is natural to ask how well a uniformly random unit vector $x$ performs. In the spherical setting, it is easy to show via a union bound that with high probability a vector $x \unif S^{n-1}$ achieves
\[ \innerproduct{v_i}{x} \leq O\left(\sqrt{\frac{\ln m}{n}}\right) \]
for all $i$. We now compare this to the boolean version of the problem, and describe the improvements in that domain.

In the most general boolean discrepancy problem, we are given an $m$-by-$n$ matrix $A$, and the goal is to minimize $\norm{Ax}_\infty$ for $x \in \{\pm 1\}^n$. Representing the rows of the matrix by $v_i$, this is equivalent to minimizing the largest inner product $\abs{\innerproduct{x}{v_i}}$. A particularly well-studied case of the boolean discrepancy problem assumes that $A$ is the incidence matrix of a set system i.e. the entries are either from $\{-1,+1\}$ or $\{0,1\}$. In this case picking a uniform $x~\unif~\{-1, +1\}^n$ is enough to achieve $\abs{\innerproduct{v_i}{x}} \leq O\left(\sqrt{n \log m}\right)$ for all $i$.
A fundamental result of discrepancy theory is Spencer's ``six standard deviations suffice'' theorem, which says that this can be asymptotically improved to  $O\left(\sqrt{n\log\frac{m}{n}}\right)$,

\begin{thm}[\cite{spencer85}]
 Given $v_1, \dots, v_m \in \{0, 1\}^n$ with $m \geq n$, there is $x \in \{-1, +1\}^n$ such that 
\[\abs{\innerproduct{v_i}{x}} \leq 6\sqrt{n \log \left(\frac{m}{n} + 1\right)}  \]
for all $i$.
\end{thm}


A famous open problem in discrepancy theory, the Koml{\'o}s conjecture, states that (in the case $m=n$) a similar bound holds if we relax the set system assumption to an assumption that the column norms of $A$ are at most $\sqrt{n}$ (or after normalization, at most 1),
\begin{restatable}{conjecture}{komlosConjecture}
Let $w_1, \dots, w_n \in \R^n$ be vectors with $\norm{w_i}_2 \leq 1$, and arrange the $w_i$ as columns of a matrix $W$. There is a boolean vector $x~\in~\{+1, -1\}^n$ so that $\norm{Wx}_\infty = O(1)$.
\end{restatable}
When the row norms of $A$ are also assumed to have norm at most 1 (which is often without much loss of generality, see Section~\ref{sec:sphere-komlos}), this is the setting in which \textsc{Spherical Discrepancy} is a relaxation of the boolean discrepancy problem. Furthermore, we will work with the ``one-sided'' discrepancy $\ip{v_i}{x}$ rather than using the absolute value of the inner product. This difference is mostly inconsequential as one can throw in $-v_i$ to bound the absolute value. The absence of absolute value corresponds more naturally to the geometric questions we consider in Section~\ref{sec:covering}.

    Spencer's argument is evidently an improvement on the random bound in the case $m=O(n)$. However, if $m \geq n^{1+c}$ for some $c>0$, then the qualitative bound $O(\sqrt{n \log \frac{m}{n}})$ is equivalent to $O(\sqrt{n\log m})$ -- the same as in the random case. Thus one must look into the actual constants and error terms to see if Spencer's bound beats the random bound whenever $m$ is significantly superlinear in $n$. In fact, the precise constant in the random argument is $\sqrt{2n\ln m}\cdot(1+o(1))$, which is definitely stronger than $6\sqrt{n\log\frac{m}{n}}$ when $m$ is superpolynomial in $n$; we are not sure the exact regime of $m$ in which Spencer's analysis (or any of the follow-up reproofs of the result) beats the random argument. These issues of constant factors will plague us in the spherical case, where we must prove tight results all the way up to $m = 2^{\sqrt{n}}$.

Spencer's theorem was only recently made algorithmic, and in the last few years there has been a spate of recent activity on algorithmic solutions to problems in discrepancy theory~\cite{DiscrepancySDPSurvey, LovettMeka, AlgorithmicWeakKomlos, BanaszczykBlues, RothvossMW}. Our algorithm for \textsc{Spherical Discrepancy} is based on a deterministic algorithm for Spencer's theorem due to Levy, Ramadas, and Rothvoss~\cite{RothvossMW}, which itself is a derandomization of a random walk-based algorithm of Lovett and Meka~\cite{LovettMeka}. The Lovett-Meka algorithm is part of a general class of \textit{partial coloring} algorithms for boolean discrepancy which produce a coloring $x \in \{-1,+1\}^n$ in $\log n$ rounds. In each round, half the remaining coordinates are set to $+1$ or $-1$. Our basic strategy for the \textsc{Spherical Discrepancy} problem is to take the first round only of a partial coloring algorithm, as this produces a vector with large norm; there is no need to further round it towards a corner of the boolean hypercube.

Finally, we mention vector discrepancy, a different relaxation of boolean discrepancy used by Lovasz~\cite{LovaszVecDisc}. Fix an $m$-by-$n$ matrix $A$. Instead of assigning $\{\pm 1\}$ to entries of a vector $x$, we assign unit vectors $x_i$ in some larger dimension, with the goal to minimize
\[\max_{i=1,\dots,m}\displaystyle \bigg\|\sum_{j = 1}^n A_{ij} x_i\bigg\|_2 . \]
When the dimension of $x_i$ is $n$, the problem is convex, and recently Nikolov~\cite{VectorKomlos} was able to verify the Koml{\'o}s conjecture in this setting using techniques from convex programming. In the spherical setting we have a global constraint $\sum_{i=1}^n x_i^2 = 1$ whereas vector discrepancy (like most other convex relaxation techniques) relaxes each variable independently. Practically speaking, the critical difference between the spherical and vector relaxations is that the domain of spherical discrepancy is non-convex and so we no longer have access to powerful tools such as duality.

For more background about discrepancy theory, see the book by Chazelle~\cite{ChazelleDiscrepancyBook}.

\subsection{Prior Work on Sphere Covering Lower Bounds.}

The \textit{spherical cap} with pole $v \in S^{n-1}$ and angular radius $\theta$ is the set
\[\{x \in S^{n-1} \mid \innerproduct{x}{v} \geq \cos \theta\}. \]
The normalized volume of a measurable set $C \subseteq S^{n-1}$ is
\[\vol(C) \defeq \Pr_{x \unif S^{n-1}}[ x \in C]. \]

There are two dual questions to ask about spherical caps. The \textit{packing} question asks: given $m$, what is the largest $\delta$ so that $m$ spherical caps of normalized volume $\delta$ can be arranged disjointly in $S^{n-1}$? The \textit{covering} question asks: given $m$, what is the smallest $\delta$ so that there are $m$ spherical caps of normalized volume $\delta$ which cover $S^{n-1}$? For both of these questions, a trivial \textit{volume bound} applies: if $m$ caps of normalized volume $\delta$ cover (respectively pack) $S^{n-1}$, then necessarily $\delta \geq 1/m$ (respectively $\delta \leq 1/m$). 
This could only be achieved if the caps could be arranged disjointly on the surface of the sphere, which is impossible except for the case of two hemispheres. The quantity $\delta m$ is called the density of the packing/covering.

The study of these questions originates from the study of maximum density packings/minimum density coverings of spheres in $\R^n$. The pioneering work of Rogers and coauthors~\cite{RogersSimplexBound, covering59} led to the \textit{simplex bound}, which states that the density of a packing or covering cannot beat a natural strategy based on tiling $\R^n$ with a regular simplex (note that $\R^n$ for $n \geq 3$ cannot be tiled with a regular simplex, so this bound is not tight).


In spherical space (and hyperbolic space), the simplex bound for packing was extended by B{\"o}r\"oczky~\cite{BoroczkyPacking}. 
For covering in spheres, we have very few bounds outside of the trivial volume bound, and the simplex bound for covering, stated implicitly in~\cite{covering59} and explicitly in~\cite[Conjecture 6.7.3]{BoroczkyFinitePackingAndCovering}, has remained unproven,

\begin{conjecture}
\label{conj:cap-density}
If a set of $m$ spherical caps each of normalized volume $\delta<1/2$ covers $S^{n-1}$, then
\[m\delta \geq \tau_{n, \delta} \]
\end{conjecture}
where $\tau_{n, \delta}$ is defined for the interested reader in Section~\ref{sec:questions}. For $\delta$ sufficiently small, the conjectured density is $\tau_{n,\delta} \approx \frac{n}{e\sqrt{e}}$ i.e. a factor of $\Omega(n)$ higher than the trivial volume bound.\footnote{We were unable to find a reference for the seemingly obvious fact that $\tau_{n,\delta} \geq c \cdot n$ for some positive constant $c$. See Conjecture~\ref{conj:density-linearity}.}
It should be noted that in general there is no perfect relationship between packing and covering, for example taking an optimal packing and extending the caps just enough to cover the whole of $S^{n-1}$ will not result in an optimal covering; there are coverings with smaller density~\cite{NonoptimalityOfPackingToCovering}. In general much more is known about packings than coverings across Euclidean, spherical, and hyperbolic spaces.



Conjecture~\ref{conj:cap-density} has been shown to hold for $n~=~3$~\cite[Theorem 5.1.1]{BoroczkyFinitePackingAndCovering}. Conjecture~\ref{conj:cap-density} is tight when $S^{n-1}$ can be tiled with regular spherical simplices using $m$ vertices, corresponding to the projection to $S^{n-1}$ of an $n$-dimensional $\{3,3,\cdots, p\}$ Coxeter polytope such as the regular simplex ($m=n+1$) and the cross-polytope ($m=2n$), as observed by Coxeter~\cite{CoxeterKissing}.


Conjecture~\ref{conj:cap-density} has also been confirmed in the regime where $\delta$ is small enough that the caps have angular radius $\theta \leq \frac{1}{\sqrt{n}}$. In this regime, the caps are small enough that $S^{n}$ looks very similar to $\R^n$, and the lower bound techniques used by Coxeter-Rogers-Few~\cite{covering59} in $\R^n$ are enough to verify Conjecture~\ref{conj:cap-density}. The authors of \cite{covering59} note that their technique can be extended to $S^n$, but they do not analyze the full range of $\delta$ for which their proof goes through. Their result only holds when $\theta \leq 1/\sqrt{n}$; see~\cite[Lemma 6.8.4]{BoroczkyFinitePackingAndCovering} for an exposition of the proof.


There is also a nontrivial covering lower bound in the regime where $\delta \geq 1/n$, which corresponds to caps of angular radius $\frac{\pi}{2} - \Theta(n^{-1/2})$. The lower bound comes from the Lusternik-Schnirelmann theorem,
\begin{thm}
If $S^{n-1}$ is covered by $n$ open or closed sets, then one of those sets contains a pair of antipodal points.
\end{thm}

Since a spherical cap with any volume $\delta < 1/2$ does not contain antipodal points, any cover must use at least $m \geq n+1$ caps. When $\delta = \Omega(1)$ this translates to a linear lower bound on the density $m\delta = \Omega(n)$, which matches Conjecture~\ref{conj:cap-density} up to a constant.
In the 60 years since \cite{covering59}, we are aware of no covering lower bounds better than the trivial $m\delta \geq 1$ outside of these two extremes. 

On the flip side, there are constructions of spherical cap covers that nearly match Conjecture~\ref{conj:cap-density},

\begin{thm}[\cite{BoroczkyEqualCovering}]
For any $0 < \varphi \leq \arccos \frac{1}{\sqrt{n}}$, there is an arrangement of spherical balls with radius $\varphi$ with density at most
\[c \cdot n\ln(1 + (n-1)\cos^2 \varphi) \]
where $c$ is an absolute constant.
\end{thm}

Translating the assumption on $\varphi$ to the parameterization by $\delta$, the upper bound on $\varphi$ says that $\delta$ is less than some fixed constant. The density bound achieved by the theorem is $O(n)$ for $\delta = \Theta(1)$ and increases to $O(n \log n)$ for constant $\theta$. The authors provide another quasilinear bound with a tighter constant,
\begin{thm}[\cite{BoroczkyEqualCovering}]
For any $\theta < \frac{\pi}{2}$, there is a covering by spherical caps of radius $\theta$ with density at most
\[ n \ln n + c\cdot n\ln \ln n\]
where $c$ is an absolute constant.
\end{thm}

We also point out the relationship with hitting sets for spherical caps. A set of points $P \subset S^{n-1}$ is a \textit{hitting set} for spherical caps of volume $\delta$ if for every spherical cap $C$ with volume $\delta$, we have $P \cap C \neq \0$. Observe that a hitting set is the same as a cover by spherical caps. The constructions above yield hitting sets for caps of volume $\delta$ of size $\widetilde{O}(n/\delta)$ --- however, they are randomized. Rabani and Shpilka~\cite{RabaniShpilkaHittingSets} show how to deterministically construct a hitting set of size $\poly(n, 1/\delta)$. Their construction works only in the large cap regime, $\delta \geq 2^{-\Omega(\sqrt{n})}$.

\subsection{Our Results and Organization of the Paper.}

In Section~\ref{sec:mw} we give an algorithm for \textsc{Spherical Discrepancy} which, analogously to Spencer's theorem, improves upon the worst-case guarantee of the random algorithm,

\begin{restatable}{thm}{mwCapDisc}
\label{thm:ip-minimization}
Let $m \geq 16n$ and $v_1,v_2, \dots,v_{m} \in S^{n-1}$ be unit vectors. We can find a vector $x\in S^{n-1}$
such that 
\[\innerproduct{v_i}{x} \leq\sqrt{\frac{2\ln\frac{m}{n}}{n}}\cdot\left(1+O\left(\frac{1}{\log{\frac{m}{n}}}\right)\right)\]
for all $i$, via a deterministic algorithm that runs in time polynomial in $m$ and $n$.
\end{restatable}

The algorithm behind Theorem~\ref{thm:ip-minimization} is based on a recent deterministic multiplicative weights algorithm for boolean discrepancy due to Levy, Ramadas, and Rothvoss~\cite{RothvossMW}. The crux of the above theorem is that we are able to achieve the optimal constant (but not optimal error term), for a large range of $m$ up to essentially $2^{\sqrt{n}}$,

\begin{restatable}{thm}{ipUpperBound}
\label{thm:ipUpperBound}
For every choice of $2^{-o(\sqrt{n})} < \delta < \frac{1}{n^2}$, there is a parameter $1/ \delta \leq m \leq \frac{2n \ln n}{\delta}$ and unit vectors $v_1, \dots, v_m \in S^{n-1}$ such that, for any $x \in S^{n-1}$ there is a $v_i$ with
\[\innerproduct{x}{v_i} \geq \sqrt{\frac{2\ln\frac{m}{n}}{n}}\cdot\left(1-o\left(\frac{1}{\sqrt{\log{\frac{m}{n}}}}\right)\right).\]
\end{restatable}

The proof of Theorem~\ref{thm:ipUpperBound} is in Section~\ref{sec:covering}. In that section, we work out the consequences of the relationship between the \textsc{Spherical Discrepancy} problem and covering the sphere by spherical caps. Our primary result is an application of Theorem~\ref{thm:ip-minimization} to prove the simplex bound for spherical space (Conjecture~\ref{conj:cap-density}) up to a log factor for caps that are not too small,\footnote{The ``log factor'' becomes $n^{1/4}$ when $m$ is near $2^{\sqrt{n}}$.}

\begin{restatable}{thm}{largeCapDensity}
\label{thm:large-cap-density}
If a set of $m$ spherical caps each of normalized volume $2^{-o(\sqrt{n})} \leq \delta < \frac{1}{2}$ covers $S^{n-1}$, then
\[\delta \geq \Omega\left(\frac{n}{m\sqrt{\log(1 + \frac{m}{n})}}\right). \]
\end{restatable}

We are able to beat the trivial volume lower bound $m\delta \geq 1$ in this theorem exactly when the algorithm behind Theorem~\ref{thm:ip-minimization} beats the trivial random bound. This requires us to carefully optimize the performance of the algorithm. 
The lower bound is algorithmic in the sense that, given a list of the $m$ caps of volume $\delta$ that does not meet the bound, there is a poly($n, m$) time algorithm that finds a point outside all caps. The lower bound assumes that the spherical caps have normalized volume at least $2^{-\sqrt{n}}$, a setting we term the \textit{large cap regime}. We explain in Section~\ref{sec:covering} how the assumption $\delta \geq 2^{-\sqrt{n}}$ naturally corresponds to a geometric property of this regime of $\delta$, namely that caps of this volume have approximately the same spherical measure as the Gaussian measure of their corresponding halfspaces. When this approximation holds, the lower bound follows as a corollary of the following covering lower bound for Gaussian space, which itself follows naturally from Theorem~\ref{thm:ip-minimization},

\begin{restatable}{thm}{gaussianDensity}
\label{thm:gaussian-density}
If $m$ halfspaces of Gaussian measure $\delta<\frac{1}{2}$ cover the $\sqrt{n}$-sphere in $\R^n$, then
\[ \delta \geq \Omega\left(\frac{n}{m\sqrt{\log(1 + \frac{m}{n})}}\right).\]
\end{restatable}
The details of the proof are found in Section~\ref{sec:covering}. Finally, in Section~\ref{sec:covering} we also study the performance of the following algorithm for generating sphere packings: run the algorithm behind Theorem~\ref{thm:ip-minimization} for $m$ iterations. We show

\begin{restatable}{thm}{generatePacking}
\label{thm:generate-sphere-packing}
Let $16n \leq m \leq 2^{o(\sqrt{n})}$. There is a deterministic algorithm which runs in time $\poly(m, n)$ that outputs $m$ points $v_1, v_2, \dots, v_m \in S^{n-1}$ such that
\[\innerproduct{v_i}{v_j} \leq\sqrt{\frac{2\ln\frac{m}{n}}{n}}\cdot\left(1+O\left(\frac{1}{\log{\frac{m}{n}}}\right)\right)\]
for all pairs $i, j$. Taking the maximal radius $r$ such that spherical caps of radius $r$ around the $v_i$ are disjoint produces a packing with density $\Omega\left(\frac{n}{2^n \sqrt{\log\frac{m}{n}}}\right)$.
\end{restatable}

In Section~\ref{sec:sphere-komlos} we show that the natural relaxation of the Koml{\'o}s conjecture for spherical discrepancy is true,
\begin{restatable}{thm}{sphericalKomlos}
\label{thm:spherical-komlos}
Let $w_1, \dots, w_n \in \R^n$ be vectors with $\norm{w_i}_2 \leq 1$, and let $W$ be the matrix with the $w_i$ as columns. Then we can find a unit vector $x \in S^{n-1}$ such that
\[\norm{Wx}_\infty = O\left(\frac{1}{\sqrt{n}}\right) \]
in time polynomial in $n$.
\end{restatable}
The algorithm in Section~\ref{sec:mw} can be modified to show this, and we explain how to do this, though for the proof as pointed out by the anonymous referee it suffices to simply cite the partial coloring lemma from boolean discrepancy.

In Section~\ref{sec:hardness-proof} we show that the \textsc{Spherical Discrepancy} problem is APX-hard via a gap-preserving reduction from NAE-3-SAT,

\begin{restatable}{thm}{apxHardness}
\label{thm:hardness-of-apx}
    There is a constant $C > 1 $ so that it is NP-hard to distinguish between instances of \textsc{Spherical Discrepancy} with $m =O(n)$ with value at most $\frac{1}{\sqrt{n}}$, and instances with value at least $\frac{C}{\sqrt{n}}$.
\end{restatable}
Lastly, in Section~\ref{sec:questions} we conclude with some further questions.

All asymptotic notation in this paper is in terms of $n$, while different parameters such as $m$ and $\delta$ are functions of $n$. The functions $\exp$ and $\ln$ denote base $e$, whereas $\log$ is used to denote base 2.

\section{Multiplicative Weights for \textsc{Spherical Discrepancy}.}
\label{sec:mw}


In this section we develop the algorithm for Theorem~\ref{thm:ip-minimization},

\mwCapDisc*

It is critical in our applications to have both the small error term and the constant $\sqrt{2}$ in the bound of Theorem~\ref{thm:ip-minimization}. In the setting of boolean discrepancy, the partial coloring method introduces a flexible constant and so it is difficult in that setting to achieve the correct constant. In fact, as pointed out by the anonymous reviewer, an application of the partial coloring lemma can be used to prove Theorem~\ref{thm:ip-minimization} with the weaker bound $O\left(\sqrt{\frac{\log \frac{m}{n}}{n}}\right)$. We do not give the details though it essentially matches the argument in Section~\ref{sec:sphere-komlos}. In the spherical setting, we are able to skirt partial coloring because the set of ``colorings'' in our domain is $S^{n-1}$, which can be rounded to by simply normalizing a ``candidate coloring''.

The algorithm behind Theorem~\ref{thm:ip-minimization} is based on a recent deterministic multiplicative weights-based algorithm for boolean discrepancy due to Levy, Ramadas, and Rothvoss~\cite{RothvossMW}. In the general setting of the multiplicative weight update (MWU) method, we are trying to make repeated decisions using the opinions of a collection of ``experts''; some of the experts tell us good advice whereas others may not, and we must learn a combination of the experts which collectively leads to a wise decision. Though the analogy is not perfect, in our case the decision is a small update direction for a candidate coloring $x$. Multiplicative weights can often be used to obtain an upper bound on the worst-case number of ``mistakes'' from following any of the experts, and in our case that corresponds to a bound on the worst-case inner product with any of the $v_i$. For more on multiplicative weights, see the survey paper~\cite{MWSurvey}. 

Similarly to their algorithm, we slowly grow a candidate coloring $x$ which is initially zero, maintaining a multiplicative weight for each input vector $v_i$ to help choose the next update direction. The weights $w_i$ essentially equal $\exp(\innerproduct{x}{v_i})$ (Lemma~\ref{lem:mw-weights}). We choose an update direction for $x$ from an eigenspace corresponding to a  small eigenvalue of the matrix $\sum_{i=1}^m w_i vv^\top$, subject to lying in a particular subspace of (linear) constraints. Ideally, we would like to move in the smallest eigenspace of this matrix. For example, consider the case when the vectors $v_i$ are mutually orthogonal. If at each timestep we make a small update in the direction of the $v_i$ with the current smallest weight, and then update the weights, this will ensure all weights remain approximately equal. Stopping the algorithm when $\norm{x}_2 = 1$ will produce $x$ with the optimal $\frac{1}{\sqrt{n}}$ inner product with each vector. However, in general it is impossible to move in the smallest eigenspace and also the subspace of constraints. Improving on~\cite{RothvossMW}, we are able to obey the constraints yet move in a direction that is essentially no worse than the average eigenvalue (Lemma~\ref{lem:mw-evs}). Under the heuristic assumption that the weights remain somewhat balanced during the run of the algorithm, this is nearly as good as moving in the smallest eigenspace.\footnote{The weights won't remain somewhat balanced, because for example two $v_i$ may point in opposite directions. But we hope that the ``critical'' weights will be somewhat balanced.}

The proof of correctness also follows the lines of Levy, Ramadas, and Rothvoss, though few technical details in the algorithm and the proof remain the same. We must do a careful induction to show that our algorithm is well-defined, and a careful analysis and choice of parameters to ensure it meets Theorem~\ref{thm:ip-minimization}.

\subsection{The Algorithm.}

Let $v_{1},\dots,v_{m}$ be unit vectors in $\R^n$. Note that the guarantee of the theorem is trivial for $m \geq n \exp(n/2)$. If the input $m$ is this large, we return any unit vector. For the remainder of this algorithm we assume that $m \leq n \exp(n/2)$.

For a PSD matrix $M\in\mathbb{R}^{n\times n},$ denote the eigenvalues of $M$ by $\mu_1(M) \geq \mu_2(M) \geq \cdots \geq \mu_n(M)$ with corresponding eigenvectors $u_i(M)$ to $\mu_i(M)$. We denote by $S^\perp$ the orthogonal complement of the span of a set of vectors $S$.

We will use parameters $\lambda, \rho, \delta,$ and $T$.
\begin{algorithm}
\label{alg:ip-min}
MWU for Spherical Discrepancy\\
\textbf{Input:} unit vectors $v_1, \dots, v_m \in S^{n-1}$\\
$x^{(0)} \gets 0^n$\\
$w_i^{(0)} \gets \exp(-\lambda^2)$\\
\textbf{for} $t = 0,\dots, T$:\\
\ftab $I^{(t)} \gets \{i \in [m] : w_i^{(t)} \geq 2\}$\\
\ftab $M^{(t)} \gets \sum_{i \notin I} w_i^{(t)}v_iv_i^\T$\\
\ftab $P^{(t)} \gets \{x^{(t)}\} \cup
\{\sum_{i\in[m]}w_i^{(t)}v_i\} \cup 
\{v_i : i \in I^{(t)}\}\cup 
\left\{ u_{j}\left(M^{(t)}\right)\ :\ 1 \leq j\leq n - \abs{I^{(t)}} - 3\right\}$
\\
\ftab $y^{(t)} \gets $any unit vector in $P^{(t)^\perp}$\\
\ftab $x^{(t+1)} \gets x^{(t)} + \delta y^{(t)}$\\
\ftab \textbf{for} $i = 1, \dots, m$:\\
\ftab \ftab $w_i^{(t+1)} \gets w_i^{(t)} \cdot \exp(\lambda \innerproduct{v_i}{\delta y^{(t)}})\cdot \rho$\\
\textbf{return} $x^{(T)} / \norm{x^{(T)}}_2$\\
\end{algorithm}

Description of the parameters: 
\begin{itemize}
    \item $\lambda = \sqrt{\ln \frac{m}{n}}$
    \item $\delta = \frac{1}{n^3}$ is the step size for updates to $x^{(t)}$
    \item $T = \frac{2(n-5)}{\delta^2}$ is the number of iterations
    \item $\rho = \exp\left(-\frac{\delta^{2}\lambda^{2}}{2(n-5)} \cdot (1 + \lambda \delta n)\right)$ is the discount factor in the weight update step. $\rho$ is very slightly less than 1. 
\end{itemize}
The initial weights $w_i^{(0)}$ don't affect the output of the algorithm so long as they're uniform.

\subsection{Runtime Analysis}
The vectors $v_i$ should be specified to $2\log n$ bits of precision so that the error in $\innerproduct{v_i}{x}$ can be incorporated into the error term of Theorem~\ref{thm:ip-minimization}.

In addition, all numerical calculations should be truncated to $30\log n$ bits of precision. The statements of all lemmas can be modified to include a polynomially small error term, which ultimately does not affect the statement of Theorem~\ref{thm:ip-minimization}. As an example, though the value $\rho$ is close to 1, it has only $\Theta(\log n)$ zeros after the decimal point, and truncating to $30\log n$ bits is enough to approximate it throughout the $T = O(n^7)$ iterations of the algorithm.

Each iteration takes $O(mn^2 + n^3)$ time: evaluating $M$ takes time $mn^2$, and computing an eigendecomposition of $M^{(t)}$ can be done in time $O(n^3)$. There are $T = O(n^7)$ iterations, for an overall runtime of $O(mn^{9} + n^{10})$. 

It seems likely that fewer iterations $T$ are sufficient (we must set the corresponding $\delta = \sqrt{\frac{2(n-5)}{T}}$). Matrix multiplication methods could also be used to lower the exponent in the time needed to compute an eigendecomposition, but we don't optimize the runtime here.

\subsection{Bounding the Maximum Inner Product}

The analysis of the algorithm will use the potential function $\Phi(t)$,
\[ \Phi(t) \defeq \sum_{i=1}^m w_i^{(t)}.\]
Initially, the potential function is $\Phi(0) = \sum_{i=1}^m \exp(-\lambda^2) = \sum_{i=1}^m \frac{n}{m} = n$.

We choose the discount factor $\rho$ so that $\Phi(t+1) \leq \Phi(t)$ for every $t$.
That is, after setting $x^{(t+1)}~\leftarrow~x^{(t)}~+~\delta y^{(t)},$
each weight is increased by a factor proportional
to $\exp\left(\lambda\cdot\delta\left\langle v_{i},y^{(t)}\right\rangle \right),$ which seems like it could increase the potential;
$\rho$ is chosen just small enough to counteract the increase. The key to producing a tight bound on $\left\langle v_{i},x^{(T)}\right\rangle $
lies in maximizing $\rho$ while still ensuring that the potential is decreasing.

As written, it is not clear that the algorithm is well-defined; it is a priori possible that the space $P^{(t)^\T}$ is trivial and does not contain a unit vector. We say that the algorithm \textit{succeeds up to time $t$} if the following conditions occur:
\begin{enumerate}[(i)]
    \item $P^{(t_0)^\perp} \neq 0$ for every $t_0 < t$ (and therefore all weights $w_i^{(t_0)}$ are properly defined as are $w_i^{(t)}$).
    \item The potential function $\Phi$ is nonincreasing from time 0 to time $t$.
\end{enumerate}

We now prove that if the algorithm succeeds up to time $t$, it succeeds up to time $t+1$, and so by induction the algorithm succeeds up to time $T$. 

\begin{lem}
If $\abs{I^{(t)}} \leq n - 3$, then $P^{(t)^\perp}\neq 0$.
\begin{proof}
We need to check that $\dim(\linspan(P^{(t)})) < n$. Of the four sets composing $P^{(t)}$, the first three contain at most $\abs{I^{(t)}}+2$ vectors. Under the assumption on $\abs{I^{(t)}}$ the fourth set contains at most $n - \abs{I^{(t}}-3$ vectors and so the result follows.
\end{proof}
\end{lem}

If the potential function is nonincreasing up to time $t$, then $\Phi(t) \leq n$. Therefore, on iteration $t+1$ the set $I^{(t)}$ has size at most $\frac{n}{2}$, and by the lemma just above $P^{(t)^\perp}$ contains a unit vector and the next set of weights will be well-defined. This fulfills the first condition for success up to time $t+1$. It remains to work towards the second condition. We at least know that $\Phi(t+1)$ is well-defined because the weights $w_i^{(t+1)}$ are well-defined --- we just need to prove $\Phi(t+1) \leq \Phi(t)$.



\begin{lem}
\label{lem:mw-evs}
For each unit vector $y\in P^{(t)^\perp},$ one has $y^{\top}M^{(t)}y\leq\frac{\Phi\left(t\right)}{n-5}$\end{lem}
\begin{proof}
Because of the fourth set composing $P^{(t)}$, $y$ is perpendicular to many eigenvalues of $M^{(t)}$, so we bound the max eigenvalue of the remaining eigenspaces. Recall that 
\[
M^{(t)}=\sum_{i \notin I^{(t)}}w_{i}^{(t)}v_{i}v_{i}^{\top}.
\]
 For any $i\in\left[m\right],$ we have $\tr{v_{i}v_{i}^{\top}}=\norm{v_{i}}^{2}=1$,
thus 
\begin{align*}
\tr{M^{(t)}} = \sum_{i \not\in I^{(t)}}w_{i}^{(t)}
= \sum_{i=1}^m w_i^{(t)} - \sum_{i \in I^{(t)}} w_i^{(t)}  \leq \Phi\left(t\right) - 2\abs{I^{(t)}}.
\end{align*}
On the other hand $\tr{M^{(t)}}$ is the sum of the eigenvalues of $M^{(t)}$. By a use of Markov's inequality, the $(n - \abs{I^{(t)}} - 2)^{th}$ largest eigenvalue of $M^{(t)}$ is at most $\frac{\Phi(t) - 2\abs{I^{(t)}}}{n - \abs{I^{(t)}} - 2}.$ Symbolically,
\begin{equation*}
    \begin{split}
        \mu_{n - \abs{I^{(t)}} - 2} &\leq \frac{\Phi\left(t\right) - 2\abs{I^{(t)}}}{n - \abs{I^{(t)}} - 2}
        = \frac{\Phi(t)}{n-5} \cdot \frac{(n-5)\left(1 - \frac{\abs{I^{(t)}}}{n}\right)}{n - \abs{I^{(t)}} - 2}\cdot \frac{1 - \frac{2\abs{I^{(t)}}}{\Phi(t)}}{1 - \frac{\abs{I^{(t)}}}{n}} 
    \end{split}
\end{equation*}

From the inductive assumption that $\Phi(t)$ is nonincreasing, $\Phi(t) \leq \Phi(0) = n$, and therefore the right term is bounded by 1. The middle term is
\[ \frac{n - \abs{I^{(t)}} - 5 + \frac{5\abs{I^{(t)}}}{n}}{n - \abs{I^{(t)}} - 2}\]
The bound $\abs{I^{(t)}} \leq n/2$ shows that this is also at most 1. Hence $\mu_{n - \abs{I^{(t)}} - 2} \leq \frac{\Phi(t)}{n-5}$. Since $y$ is a linear combination of eigenvectors with eigenvalues at most this value, we have $y^{\top}M^{(t)}y~\leq~\frac{\Phi\left(t\right)}{n-5}$.
\end{proof}

\begin{restatable}{proposition}{expIneq}
\label{lemma:exp-inequality}
For any $0 \leq x \leq 1$,
\[ e^x \leq 1 + x + \frac{x^2}{2} + \frac{x^3}{2}.\]
\end{restatable}
\begin{proof}
See Appendix A.
\end{proof}


\begin{lem}
$\Phi\left(t+1\right)\leq\Phi(t)$.
\end{lem}
\begin{proof}
The recursive update for the weights at
time $t+1$ is given in the algorithm,
\begin{align*}
\Phi(t+1) & = \sum_{i=1}^m w_{i}^{(t+1)} = \sum_{i=1}^m w_{i}^{(t)}\cdot\exp\left(\lambda\delta\left\langle v_{i},y^{(t)}\right\rangle \right)\cdot\rho.
\end{align*}
We assumed that $m \leq n \exp(n/2)$ so that $\lambda \leq \sqrt{n/2}$ and $\lambda\delta\langle v_i, y^{(t)}\rangle \leq \lambda\delta \leq 1$, and we can apply the previous Lemma,


\begin{align*}
 \Phi(t+1) / \rho&  \leq \sum_{i=1}^{m}w_{i}^{(t)}\left( {\vphantom{\frac{\lambda^{2}\delta^{2}\left\langle v_{i},y^{(t)}\right\rangle ^{2}}{2} + \frac{\lambda^{3}\delta^{3}\left\langle v_{i},y^{(t)}\right\rangle ^{3}}{2}}} 1+\lambda\delta\left\langle v_{i},y^{(t)}\right\rangle \right. \left. +\frac{\lambda^{2}\delta^{2}\left\langle v_{i},y^{(t)}\right\rangle ^{2}}{2} + \frac{\lambda^{3}\delta^{3}\left\langle v_{i},y^{(t)}\right\rangle ^{3}}{2}\right)\\
 & =  \sum_{i=1}^{m}w_{i}^{(t)}+\lambda\delta \left \langle \sum_{i=1}^{m}w_{i}^{(t)}v_{i},y^{(t)}\right\rangle +\frac{\lambda^{2}\delta^{2}}{2}\sum w_{i}^{(t)}\left\langle v_{i},y^{(t)}\right\rangle ^{2} + \frac{\lambda^{3}\delta^{3}}{2}\sum w_{i}^{(t)}\left\langle v_{i},y^{(t)}\right\rangle ^{3}.
 \end{align*}
 Since $y^{(t)} \in P^{(t)^{\perp}}$, the second term is zero. Furthermore, $\sum_{i=1}^{m} w_{i}^{(t)}\left\langle v_{i},y^{(t)}\right\rangle ^{2} =  y^{(t)\top}M^{(t)}y^{(t)}$, so by Lemma~\ref{lem:mw-evs} we have,
 \begin{align*}
  \Phi(t+1) / \rho & \leq  \Phi(t)+\frac{\lambda^{2}\delta^{2}}{2(n-5)}\Phi(t) + \frac{\lambda^{3}\delta^{3}}{2}\sum w_{i}^{(t)}\underbrace{\left\langle v_{i},y^{(t)}\right\rangle ^{3}}_{\leq 1}\\
& \leq \Phi(t)+\frac{\lambda^{2}\delta^{2}}{2(n-5)}\Phi(t) + \frac{\lambda^{3}\delta^{3}}{2}\Phi(t)\\
& = \Phi(t) \cdot \left(1+\frac{\lambda^{2}\delta^{2}}{2(n-5)}(1 + \lambda\delta n)\right)\\
\Phi(t+1)/\rho & \leq \Phi(t) \cdot \exp\left(\frac{\lambda^{2}\delta^{2}}{2(n-5)}(1 + \lambda\delta n)\right).
\end{align*}

In the last line, we use the inequality $1+x\leq e^{x}$. The exponential term is exactly $1/\rho$, therefore we conclude $\Phi(t+1) \leq \Phi(t)$. 
\end{proof}

This finishes the proof by induction that the algorithm succeeds up to time $T$. We now proceed to bound the max inner product $\innerproduct{v_i}{x^{(T)}}$, starting with a few lemmas.

\begin{lem}
\label{lem:mw-weights}
$w_{i}^{(t)}=\exp\left(\lambda\left\langle v_{i}, x^{(t)}\right\rangle -\lambda^{2} \right) \cdot \rho^t$
\end{lem}
\begin{proof}
The weights are initially $\exp(-\lambda^2)$. Each iteration they are multiplied by a factor of $\rho$, and also by $\exp(\lambda\delta \innerproduct{v_i}{y^{(t)}})$, so
\[w_i^{(t)} = \exp\left(\lambda\left\langle v_{i}, \sum_{t'< t} \delta y^{(t')}\right\rangle -\lambda^{2} \right) \cdot \rho^t. \]
On the other hand, $x^{(t)} = \sum_{t' < t} \delta y^{(t')}$.
\end{proof}

\begin{lem}
\label{lem:mw-norm}
$\Vert x^{(t)}\Vert_{2}=\delta\sqrt{t}$\end{lem}
\begin{proof}
Note that $x^{(t)}=x^{(t-1)}+\delta y^{(t-1)},$ and since $y^{(t)}$ is a unit vector orthogonal to $x^{(t)}$ due to the subspace $U_1$,
\begin{eqnarray*}
\Vert x^{(t)}\Vert_{2}^{2} & = & \Vert x^{(t-1)}\Vert_{2}^{2}+\delta^{2}\Vert y^{(t-1)}\Vert_{2}^{2}=\Vert x^{(t-1)}\Vert_{2}^{2}+\delta^{2}.
\end{eqnarray*}
Since $x^{(0)}={\bf 0}^{n},$ the result follows.
\end{proof}

\begin{lem}
\label{lem:mw-wt-bd}
At all times $t$, the weights satisfy $\max_i w_i^{(t)} \leq 3$.
\end{lem}
\begin{proof}
Once a weight becomes greater than 2, it moves into the set $I^{(t)}$ and further moves are orthogonal to $v_i$, meaning the weight only decreases due to $\rho$. On the update that the weight moves into $I^{(t)}$,
\begin{align*}
    w_i^{(t)} & = w_i^{(t-1)} \cdot \exp(\lambda\cdot \delta \innerproduct{v_i}{y^{(t-1)}}) \cdot \rho \\
    & \leq 2 \cdot \exp(\lambda \cdot \delta) \cdot \rho\\
    & \leq 2 \cdot \exp(\lambda \cdot \delta) \leq 3
\end{align*}
Note that the weight may move out of and back into $I^{(t)}$ in the future due to the discount factor $\rho$.
\end{proof}

We are now in a position to prove Theorem~\ref{thm:ip-minimization}. We will take $x$ to be $x^{(T)} / \norm{x^{(T)}}$. 
\begin{lem}
\label{lem:mw-final}
$\left\langle v_{i}, x^{(T)}/\norm{x^{(T)}}\right\rangle  \leq \lambda\sqrt{\frac{2}{n}} + O\left(\frac{1}{\lambda\sqrt{n}}\right) $
\end{lem}

\begin{proof}
We have an exact expression for the weights by Lemma~\ref{lem:mw-weights} and a bound on the weights by Lemma~\ref{lem:mw-wt-bd},
\begin{eqnarray*}
w_{i}^{(T)} & =\exp\left(\lambda\left\langle v_{i}, x^{(T)}\right\rangle -\lambda^{2}\right) \cdot \rho^T \leq & 3.
\end{eqnarray*}
We take logs, solve for $\innerproduct{v_i}{x^{(T)}}$, and optimize the resulting bound.
\begin{eqnarray*}
\lambda\left\langle v_{i}, x^{(T)}\right\rangle -\lambda^{2}+T \ln \rho  & \leq & \ln 3\\
\left\langle v_{i}, x^{(T)}\right\rangle  & \leq & \lambda - \frac{T\ln \rho}{\lambda} + \frac{1}{\lambda}\ln 3
\end{eqnarray*}
Plug in $\ln \rho = -\frac{\lambda^2\delta^2}{2(n-5)}\cdot(1 + \lambda\delta n)$ and normalize by $\norm{x^{(T)}} = \delta\sqrt{T}$~(Lemma~\ref{lem:mw-norm}),
\begin{align*}
\left\langle v_{i}, x^{(T)}/\norm{x^{(T)}}\right\rangle  \leq \frac{\lambda}{\delta\sqrt{T}} + \frac{\lambda \delta \sqrt{T}}{2(n-5)}\cdot(1 + \lambda \delta n) + \frac{1}{\lambda\delta\sqrt{T}}\ln 3\\
\leq \left(\frac{\lambda}{\delta\sqrt{T}} + \frac{\lambda \delta \sqrt{T}}{2(n-5)}\right)(1 + \lambda \delta n) + \frac{1}{\lambda\delta\sqrt{T}}\ln 3 .
\end{align*}
The choice of parameters $\delta \sqrt{T} = \sqrt{2(n-5)}$ balances the first two terms to the value $\frac{\lambda}{\sqrt{2(n-5)}}$.
\[\left\langle v_{i}, x^{(T)}/\norm{x^{(T)}}\right\rangle  \leq \lambda\sqrt{\frac{2}{n-5}}(1 + \lambda\delta n) + O\left(\frac{1}{\lambda\sqrt{n}}\right) \]
Recalling that $\lambda \leq \sqrt{n}$ and $\delta = \frac{1}{n^{3}}$, we can absorb the term $\lambda\delta n$ into the remaining error,
\[ = \lambda\sqrt{\frac{2}{n-5}} + O\left(\frac{1}{\lambda\sqrt{n}}\right). \]
Bounding $\sqrt{\frac{1}{n-5}} \leq \sqrt{\frac{1}{n}} + \frac{5}{n^{3/2}}$ shows we can incorporate this term into the error as well.
\end{proof}



Finally, if we plug in $\lambda = \sqrt{\log \frac{m}{n}}$ to Lemma~\ref{lem:mw-final} we recover Theorem~\ref{thm:ip-minimization}.


\section{Lower Bounds for Covering Problems.}
\label{sec:covering}

As described in the introduction, algorithms for \textsc{Spherical Discrepancy} immediately yield algorithmic lower bounds for the problem of covering the hypersphere by hyperspherical caps, which will let us show the following theorem,

\largeCapDensity*

Essentially, the proof boils down to the following calculation: set the right-hand side of Theorem~\ref{thm:ip-minimization} equal to $\cos \theta$ and solve for the normalized volume $\delta$ of a spherical cap with angular radius $\theta$. The calculation is simplified by noticing that Theorem~\ref{thm:ip-minimization} has a form which applies perfectly to the following question about Gaussian space: how many halfspaces of Gaussian measure $\delta$ are required to cover the surface of the $\sqrt{n}$-sphere? Theorem~\ref{thm:ip-minimization} gives an algorithmic lower bound for this question for every $\delta$,

\gaussianDensity*

In order to reduce Theorem~\ref{thm:large-cap-density} to Theorem~\ref{thm:gaussian-density}, all we require is a restatement of the classical fact attributed to Poincar{\'e} (but see~\cite{SphereCoordinateHistory} for a more complete history) that the first coordinate of a uniformly random point on the sphere is approximately distributed like $\mathcal{N}(0, \frac{1}{n})$; modern formulations appear in~\cite{StamSphereUnif, SpruillSphereUnif}. However, we need a more quantitative version of the bound (Lemma~\ref{lem:measure-apx}), and we spend some work showing this.

This approach also explains the assumption of the theorem $\delta \geq 2^{-o(\sqrt{n})}$, a setting which we call the \textit{large cap regime}. When the spherical caps have at least this volume, the approximation between spherical and Gaussian space is good enough to deduce Theorem~\ref{thm:large-cap-density} from Theorem~\ref{thm:gaussian-density}. More specifically, if $C = \{x \in S^{n-1} : \ip{x}{p} \geq \cos \theta\}$ is a spherical cap, then the halfspace $H = \{x \in \R^n : \ip{x}{p} \geq \sqrt{n}\cos \theta\}$ has $\gamma(H) \approx \vol(C)$ within a constant factor.

An important regime for Conjecture~\ref{conj:cap-density} that is outside the large cap regime is when $\theta \in (0, \pi/2)$ is a constant. The corresponding $\delta$ is
\[\delta = \Theta\left(\frac{\sin^n \theta}{\sqrt{n}}\right). \]
For example, the case $\theta = \pi /6$ corresponds to the extensively-studied kissing numbers~\cite{KissingNumbersSurvey, PerkinsKissingNumbers}. For $\delta$ in this range, the approximation in Lemma~\ref{lem:measure-apx}
is off by an exponential factor, and one cannot deduce \textit{any} nontrivial spherical measure lower bound.

Note that the known spherical cap constructions show that Theorem~\ref{thm:gaussian-density} is tight up to a log factor provided $\delta \geq 2^{-o(\sqrt{n})}$. When $\delta$ is significantly smaller than $2^{-\sqrt{n}}$, Theorem~\ref{thm:gaussian-density} is not tight and the volume bound for spherical caps provides an exponential improvement.

There is a slightly more natural version of the Gaussian covering problem which unfortunately we are not able to resolve. We say that a set $S$ is a $(1-\eps)$-cover of a random variable $X$ if $\Pr[X \not\in S] \leq \eps$. We say a family of sets is a $(1-\eps)$-cover if their union is. If $m$ halfspaces of Gaussian measure $\delta$ are a $1/2$-cover of a standard Gaussian random variable, then there is a trivial volume lower bound $m\delta \geq 1/2$. We conjecture a linear density lower bound for this covering problem,

\begin{restatable}{conjecture}{gaussianDensityConjecture}
\label{conj:gaussian-density}
If $m$ halfspaces of Gaussian measure $\delta<\frac{1}{2}$ are a $1/2$-cover of a standard $n$-dimensional Gaussian random variable, then
\[ m\cdot\delta \geq \Omega\left(n\right).\]
\end{restatable}

The choice of $1/2$ corresponds to the fact that the Gaussian measure of the ball of radius $\sqrt{n}$ is approximately $1/2$. Therefore a natural strategy to $1/2$-cover the Gaussian random variable is to cover the surface of the $\sqrt{n}$-sphere, for which the bound in Theorem~\ref{thm:gaussian-density} applies. However, we are unable to rule out the existence of $1/2$-covers with smaller total density. 

\subsection{Gaussian Space Lower Bound.}

We now prove Theorem~\ref{thm:gaussian-density}. We need the following fact,
\begin{restatable}{proposition}{gaussianTail}
\label{fact:gaussian-tail}
Let $\varphi$ and $\overline{\Phi}$ denote the PDF and tail probability respectively of the standard normal distribution, 
\[\varphi(t) = \frac{1}{\sqrt{2\pi}} \exp(-t^2/2) \qquad \qquad \overline{\Phi}(t) = \int_{t}^\infty \varphi(s)\, ds\]
There is a constant $C$ so that for $t\geq 1$, $C\frac{\varphi(t)}{t} \leq \overline{\Phi}(t) \leq \frac{\varphi(t)}{t}$.
\end{restatable}
\par\noindent\textit{Proof.}
See Appendix A.

As can be seen in the proof, the reason for such a tight error bound in Theorem~\ref{thm:ip-minimization} is that the error appears in the exponent here, meaning even small error becomes amplified. Now we prove Theorem~\ref{thm:gaussian-density}.

\begin{proof}
Let $\{H_i\}$ be a set of $m$ halfspaces each of Gaussian measure $\delta$ whose union covers the $\sqrt{n}$-sphere and assume for the moment that $m \geq 16n$. Let $v_i$ be a normal unit vector to $H_i$. By Theorem~\ref{thm:ip-minimization} we can find a vector $x$ in the $\sqrt{n}$-sphere such that
\[\innerproduct{v_i}{x} \leq \sqrt{2\log\frac{m}{n}}\cdot\left(1+O\left(\frac{1}{\log{\frac{m}{n}}}\right)\right) \]
for each $i$. Let $\epsilon = O\left(\frac{1}{\log{\frac{m}{n}}}\right)$ denote the error term. 

By assumption $x$ is covered by one of the halfspaces $H_i$, expressed as
\[H_i =  \{x \in \R^n \mid \innerproduct{v_i}{x} \geq \overline{\Phi}^{-1}(\delta)\}. \]
Therefore,
\[\overline{\Phi}^{-1}(\delta) \leq \sqrt{2\log\frac{m}{n}}\cdot\left(1+\epsilon\right) \]
\[ \delta \geq \overline{\Phi}\left(\sqrt{2\log\frac{m}{n}}\cdot\left(1+\epsilon\right)\right).\]
Applying Proposition~\ref{fact:gaussian-tail},
\[\delta \geq C\exp\left(-\left(\sqrt{2\log\frac{m}{n}}\cdot(1 +\epsilon)\right)^2 / 2 \cdot \frac{(1 +\epsilon)^{-1}}{\sqrt{2\log\frac{m}{n}}}\right)\]
\[= \Theta\left(\left(\frac{n}{m}\right)^{1 + O(\epsilon)} \cdot \frac{1}{\sqrt{\log\frac{m}{n}}} \right). \]
The component of the exponent $\epsilon = O\left(\frac{1}{\log{\frac{m}{n}}}\right)$ contributes a multiplicative constant factor, and we have the claim,
\[\delta \geq \Omega\left(\frac{n}{m\sqrt{\log (1 + \frac{m}{n})}}\right). \]

We deal with the case in which there are fewer than $16n$ halfplanes. Add vectors until we have $16n$. The theorem in this case tells us
\[\delta \geq \Omega\left(\frac{n}{16n\sqrt{\log\frac{16n}{n}}}\right) = \Omega(1)\]
By the Lusternik-Schnirelmann theorem $m \geq n$, therefore $m \delta \geq \Omega(n)$.
\end{proof}

\subsection{From Gaussian Space to the Sphere.}

Here we prove Theorem~\ref{thm:large-cap-density} from Theorem~\ref{thm:gaussian-density}. Recall our strategy: let $C$ be a spherical cap on $S^{n-1}$, say with equation
\[C = \{x \in S^{n-1} \mid \innerproduct{x}{v} \geq \cos(\theta) \} \]
and let $H$ be the halfspace with the same intersection but on the $\sqrt{n}$-sphere,
\[H = \{x \in \R^n \mid \innerproduct{x}{v} \geq \cos(\theta)\sqrt{n} \}.\]
We want $\vol(C) \approx \gamma(H)$ to deduce Theorem~\ref{thm:large-cap-density} from Theorem~\ref{thm:gaussian-density}. For angle $\theta=\pi/2$, both shapes reduce to hemispheres, and for angle $\theta \geq \pi/2 - O(n^{-1/2})$ both shapes have the same measure $\Omega(1)$ up to $o(1)$ via the classical theorem that $\sqrt{n}x_1 \overset{d}{\to} N(0,1)$ for $x \unif S^n$. The following lemma shows that the two have similar measure for angular radius at least $\pi/2 - o(n^{-1/4})$, or equivalently when one shape has measure at least $2^{-o(\sqrt{n})}$,
\begin{lem}
\label{lem:measure-apx}
Let $C \subset S^{n-1}$ be a spherical cap with angular radius $\theta$. Let
\[H = \{x \in \R^n \mid \innerproduct{x}{v} \geq \sqrt{n}\cos\theta \}\]  where $v$ is some fixed vector, let $\phi = \pi / 2 - \theta$, and assume $\phi = o(n^{-1/4})$. Then 
\[\vol(C) \sim \gamma(H) \sim \overline{\Phi}(\sqrt{n}\phi). \]
 \end{lem}



The intuition on why this is the correct assumption on $\theta$ is that both $C$ and $H$ have the same intersection with $\{-1, +1\}^n$, and for this $\theta$ both $\vol(C)$ and $\gamma(H)$ are well estimated by sampling a uniform boolean point. To explain further, the particular halfspace $H = \{ x \in \R^n \mid \innerproduct{x}{\vec{\frac{1}{\sqrt{n}}}} \geq t\}$ has Gaussian measure $\overline{\Phi}(t)$. On the other hand, the measure can be estimated by sampling $X \unif \{+1, -1\}^n$. Setting $S = \frac{1}{\sqrt{n}}\sum_{i=1}^n X_i$, by the Central Limit Theorem we expect $\Pr[S \geq t] \approx \overline{\Phi}(t)$; this is a good approximation for constant $t$, and in fact for $t$ up to $n^{1/4}$. However, after this point the tail probability becomes exponentially smaller (the subject of large deviation theory is to determine the correct exponent, which in this case is given for all $t$ by the Chernoff-Hoeffding theorem). One interpretation of the lemma is that we prove the boolean sampling procedure also accurately estimates $\vol(C)$ in the range where it estimates $\gamma(H)$, namely $t \leq n^{1/4}$. 

There is another geometric interpretation of the lemma. Let $G$ be the cone in $\R^n$ that contains $C$,
\[G = \left\{x \in \R^n :  \innerproduct{\frac{x}{\norm{x}_2}}{v} \geq \cos(\theta)\right\}. \]
Due to the rotational symmetry of the Gaussian measure, $\vol(C) = \gamma(G)$. The halfspace and the cone are two natural bodies that pass through $C$ and $\vol(C) \approx \gamma(H)$ is equivalent to asking that they have similar Gaussian measure.


\begin{proof}
The Gaussian measure of $H$ is
\begin{align*}
    \gamma(H) & = \overline{\Phi}(\sqrt{n}\cos\theta) = \overline{\Phi}(\sqrt{n}\sin\phi)\\
     & = \overline{\Phi}(\sqrt{n}\phi - O(\sqrt{n}\phi^3))  \\
     & = \overline{\Phi}(\sqrt{n}\phi) +  \int_{\sqrt{n}\phi - O(\sqrt{n}\phi^3)}^{\sqrt{n}\phi} \varphi(x)\, dx\\
     & \leq \overline{\Phi}(\sqrt{n}\phi) + O(\sqrt{n}\phi^3)\frac{e^{-(n\phi^2 - O(n\phi^4))/2}}{\sqrt{2\pi}}.
\end{align*}
We want to show that the second term is negligible compared to the first. Using the assumption that $\phi = o(n^{-1/4})$,
\[ = \overline{\Phi}(\sqrt{n}\phi) + O(\sqrt{n}\phi^3)\frac{e^{-n\phi^2/2}}{\sqrt{2\pi}}(1 + o(1)). \]
We distinguish two cases. When $\phi \geq 1/\sqrt{n}$, by Proposition~\ref{fact:gaussian-tail} we have
\[ \overline{\Phi}(\sqrt{n}\phi) = \Omega\left(\frac{e^{-n\phi^2 / 2}}{\sqrt{n}\phi}\right).\]
The ratio of the second term to the first is therefore
\[ O\left(n\phi^4\right) = o(1). \]
In the second case, $\phi < 1/\sqrt{n}$. The first term is at least $\overline{\Phi}(1) = \Omega(1)$, whereas the second term is at most $O(1/n)$. Put together, this shows $\gamma(H) = \overline{\Phi}(\sqrt{n}\phi)(1 + o(1))$.

Now we turn to $C$. A simple formula for $\vol(C)$ is given in~\cite{capVolume},
\[\vol(C) =  \frac{1}{\sqrt{\pi}} \frac{\Gamma(\frac{n}{2})}{\Gamma(\frac{n-1}{2})}\int_0^\theta \sin^{n-2}x\, d x\]
\[= \frac{\sqrt{n}}{\sqrt{2\pi}}(1+o(1)) \int_{\phi}^{\pi/2} \cos^{n-2}x \,d x\]
Making the substitution $x = y/\sqrt{n}$,
\[= \frac{1}{\sqrt{2\pi}}(1+o(1))\int_{\sqrt{n}\phi}^{\sqrt{n}\pi/2} \cos^{n-2}(y/\sqrt{n}) \,d y. \]
The integrand is approximately $\cos^n(y / \sqrt{n}) \approx (1 - y^2/2n)^n \approx e^{-y^2/2}$, so at least heuristically we have the claim $\vol(C) \approx \overline{\Phi}(\sqrt{n}\phi)$. To make this argument rigorous, let $T = \max(n^{0.2}, n^{0.7}\phi)$ and break the integral into two pieces,
\begin{align*}\frac{1}{\sqrt{2\pi}}\int_{\sqrt{n}\phi}^{\sqrt{n}\pi/2} \cos^{n-2}(y/\sqrt{n}) \,d y &=  \frac{1}{\sqrt{2\pi}}\int_{\sqrt{n}\phi}^{T} \cos^{n-2}(y/\sqrt{n}) \,d y  \\ &+ \frac{1}{\sqrt{2\pi}}\int_{T}^{\sqrt{n}\pi/2} \cos^{n-2}(y/\sqrt{n}) \,d y.\end{align*}
The first piece can be Taylor expanded as the upper bound of integration is $o(\sqrt{n})$. The integrand is
\begin{equation*}
    \begin{split}
    \cos^{n-2}\left(\frac{y}{\sqrt{n}}\right) & = \left(1 - \frac{y^2}{2n} + O\left(\frac{y^4}{n^2}\right)\right)^{n-2} \\ & = \left(1 - \frac{y^2}{2n}\right)^n\left(1 + O\left(\frac{y^4}{n^2}\right)\right)^{n} \cdot(1 + o(1))
    \end{split}
\end{equation*}
Proposition~\ref{fact:exp-fidelity}, in the appendix, shows that the first term is asymptotic to $e^{-y^2/2}$ and the second term is $1+o(1)$, with both using the assumption that $y \leq O(\sqrt{n}\phi) = o(n^{1/4})$. The first piece is therefore equal to
\[\frac{1 + o(1)}{\sqrt{2\pi}}\int_{\sqrt{n}\phi}^{T}e^{-y^2/2} \, dy. \]

The second piece is exponentially smaller than the first. However, we leave this claim unproven for a moment in order to show that the first piece is $(1 + o(1))\overline{\Phi}(\sqrt{n}\phi)$. That is, we want to show that
\[\int_{\sqrt{n}\phi}^{T}e^{-y^2/2} \, dy \geq (1 - o(1)) \int_{\sqrt{n}\phi}^{\infty}e^{-y^2/2} \, dy.\]
The tail is upper bounded by Proposition~\ref{fact:gaussian-tail},
\begin{align*}
    \int_{T}^\infty e^{-y^2/2}\, dy &\leq \frac{e^{-T^2/2}}{T}.
\end{align*}
On the other hand Proposition~\ref{fact:gaussian-tail} gives us a lower bound on $\overline{\Phi}(\sqrt{n}\phi)$, at least in the case when $\phi \geq 1/\sqrt{n}$,
\[\overline{\Phi}(\sqrt{n}\phi)  \geq C\frac{e^{-n\phi^2/2}}{\sqrt{n}\phi}.\]
The exponent $e^{-T^2/2} = O(e^{-n^{1.4}\phi^2/2})$ is exponentially smaller. If $\phi < 1/\sqrt{n}$, then $\overline{\Phi}(\sqrt{n}\phi) = \Omega(1)$, whereas the tail is bounded by $1/T = O(n^{-0.2})$. In either case, the tail is $o(\overline{\Phi}(\sqrt{n}\phi))$ as needed.

Now we return to bounding the second piece. We show that it's negligible compared to $\overline{\Phi}(\sqrt{n}\phi)$. Bounding the integrand by the left endpoint shows it's at most $\frac{\sqrt{n}\pi}{2} \cdot \cos^{n-2}(T)$
which as done above can be approximated using the Taylor expansion,
\[= (1 + o(1))\cdot \frac{\sqrt{n}\pi}{2}\cdot e^{-T^2/2}. \]
In the case that $\phi \geq 1/\sqrt{n}$, this is $O(e^{-n^{1.4}\phi^2})$, whereas by Proposition~\ref{fact:gaussian-tail}, $\overline{\Phi}(\sqrt{n}\phi)$ has exponent $e^{-n\phi^2}$, and hence this term is exponentially smaller. In the case that $\phi < 1/\sqrt{n}$, this is $O(\sqrt{n}e^{-n^{0.4}/2})$ whereas $\overline{\Phi}(\sqrt{n}\phi) = \Omega(1)$. This shows that the second piece of the integral is negligible, while the first is asymptotic to $\overline{\Phi}(\sqrt{n}\phi)$, which completes the proof of the lemma.
\end{proof}


Finally, we fill in the details of the proof of Theorem~\ref{thm:large-cap-density} using Theorem~\ref{thm:gaussian-density}.

\begin{proof}
Say we have a collection of $m$ caps $\{C_i\}$ whose union covers $S^{n-1}$, and each has angular radius $\theta$ and normalized volume $\delta$ with $2^{-o(\sqrt{n})} \leq \delta < \frac{1}{2}$. 
Let $v_i$ be the pole of cap $C_i$, and define the halfspaces
\[H_i = \{x \in \R^n \mid \innerproduct{x}{v_i} \geq \sqrt{n}\cos\theta \}. \]
The intersection of $H_i$ with the $\sqrt{n}$-sphere is exactly $\sqrt{n}C_i$, and the assumption that the $\{C_i\}$ cover $S^{n-1}$ tells us that the $H_i$ cover the $\sqrt{n}$-sphere. 
Apply Theorem~\ref{thm:gaussian-density} to the collection of $\{H_i\}$,
\[\gamma(H_i) \geq \Omega\left(\frac{n}{m \sqrt{\log(1 + \frac{m}{n})}}\right) .\]
We would like to now apply Lemma~\ref{lem:measure-apx}; to do so we need to bound $\theta$ given $\delta$. By Proposition~\ref{fact:gaussian-inv}, a halfspace of Gaussian volume $\delta$ has
\[\overline{\Phi}^{-1}(\delta) = \sqrt{2\ln 1/\delta} + o(1) = o(n^{1/4}). \]
Using Lemma~\ref{lem:measure-apx}, the spherical cap through this halfspace has measure $\delta(1+o(1))$ and angle $\frac{\pi}{2} - o(n^{-1/4})$. Enlarging the angle by $o(n^{-1/4})$ is enough to contain the spherical cap of volume $\delta$.

Now by Lemma~\ref{lem:measure-apx}, $\gamma(H_i)$ is within a constant factor of $\delta$, and therefore
\[ \delta \geq \Omega\left(\frac{n}{m \sqrt{\log(1 + \frac{m}{n})}}\right).\]
\end{proof}



\subsection{A Matching Bound for Theorem~\ref{thm:ip-minimization}.}

The lemmas established in this section can be used to convert the cap covering bounds of B{\"o}r{\"o}czky and Wintsche~\cite{BoroczkyEqualCovering} to bounds for Theorem~\ref{thm:ip-minimization}. We don't make too much fuss about the constants for covers and suffice to work with the simpler covering density bound of $2n \ln n$.

\ipUpperBound*
\begin{proof}
    Fix $\delta$. Find a set of caps of $m$ caps of volume $\delta$ with polar vectors $v_i$ that cover the sphere, with
    \[1 \leq m\delta \leq 2n \ln n. \]
    Letting $\theta$ be the angular radius of the caps, for any $x \in S^{n-1}$ there is $v_i$ with $\innerproduct{v_i}{x} \geq \cos \theta$. The calculation performed in the proof of Theorem~\ref{thm:large-cap-density} showed that $\theta$ was in the correct range for Lemma~\ref{lem:measure-apx}.
    \begin{align*}
        \overline{\Phi}(\sqrt{n}\cos\theta) &= \delta(1+o(1)) \\
        \cos \theta &= \frac{1}{\sqrt{n}}\overline{\Phi}^{-1}(\delta(1+o(1)))
    \end{align*}
    Taylor expand $\overline{\Phi}^{-1}$ via Proposition~\ref{fact:gaussian-inv}, given in the appendix,
    \begin{align*}
        \overline{\Phi}^{-1}(\delta(1+o(1))) &= \sqrt{2\ln 1/\delta} + o(1)\\
        & \geq \sqrt{2\ln \frac{m}{n} -2\ln (2\ln n)} + o(1)\\
        & = \sqrt{2\ln \frac{m}{n}} + o(1).
    \end{align*}
    Put together, we have
    \[\innerproduct{v_i}{x} \geq \sqrt{\frac{2\ln \frac{m}{n}}{n}}\left(1 + o\left(\frac{1}{\sqrt{\log \frac{m}{n}}}\right)\right). \]
\end{proof}






\subsection{Generating Sphere Packings.}

We verify the performance of Algorithm~\ref{alg:ip-min} for generating sphere packings,
\generatePacking*

The density of this packing density upper bound is related to the density of the covering lower bound by a factor of $2^n$. This is (not directly due to, but essentially) because of the well-known generic relationship between packing and covering: doubling the radius of a maximal packing produces a covering, with density scaled by $2^n$. However, this generic relationship is in general not tight.

\begin{proof}
The first part is immediate from Theorem~\ref{thm:ip-minimization}. For the second part, let $2r$ be the minimum distance between the $v_i$. Taking a cap of radius $2r$ around each point covers $S^{n-1}$, and the density of this covering was computed in  Theorem~\ref{thm:large-cap-density} to be at least $\Omega\left( \frac{n}{\sqrt{\log\frac{m}{n}}}\right)$. Halving the radius gives a disjoint set of caps of the desired density.
\end{proof}

\section{Koml{\'o}s Problem in the Spherical Domain.}
\label{sec:sphere-komlos}

    Seeing as the spherical domain sidesteps the rounding present in partial coloring, which seems to have inherent flaws when trying to prove tight discrepancy results in the boolean domain (see~\cite{BansalGargBeyondPartialColoring} for some discussion), it is natural to hope that the spherical versions of some open problems from discrepancy theory are more tractable. In particular we can adapt our Algorithm~\ref{alg:ip-min} to resolve a version of the Koml{\'o}s problem in the spherical setting,
    \begin{restatable}{thm}{sphericalKomlos}
    \label{thm:spherical-komlos}
    Let $w_1, \dots, w_n \in \R^n$ be vectors with $\norm{w_i}_2 \leq 1$, and let $W$ be the matrix with the $w_i$ as columns. Then we can find a unit vector $x \in S^{n-1}$ such that
    \[\norm{Wx}_\infty = O\left(\frac{1}{\sqrt{n}}\right) \]
    in time polynomial in $n$.
    \end{restatable}
    
    The idea is a rather general one taken from boolean discrepancy theory. A priori some of the rows of $W$ could have norm as large as $\sqrt{n}$. Even so, in the special case of the Koml{\'o}s conjecture where the rows of $W$ are restricted to have norm $O(1)$, the core of the problem still seems to be captured; for many algorithms for discrepancy problems, the row norms are assumed to be bounded without loss of generality.
    This holds for algorithms that iteratively build a solution from small update vectors subject to linear number linear number $\gamma n$ ($\gamma <1 $) of linear constraints. The bounds $\norm{w}_2 \leq 1$ ensure there are at most $\eps^2 n$ row vectors with norm greater than $1/\eps$. Thus as long as $\gamma + \eps^2 < 1$, the number of linear constraints we have is significantly less than $n$, and there is still room to find an update vector orthogonal to rows of large norm.


In fact, as pointed out by the anonymous reviewer, the above idea can be instantiated in a black box manner using just the generic partial coloring lemma  of Lovett and Meka,

    \begin{thm}[\cite{LovettMeka}]
        \label{thm:LovettMeka}
        Let $v_1, \dots, v_m \in \R^n$ be vectors, $x^{(0)} \in [-1,+1]^n$ be a starting point, and let $c_1 , \dots, c_m \geq 0$ be thresholds so that $\sum_{j=1}^m \exp(-c_j^2/16) \leq \frac{n}{32}$. Let $\delta > 0$ be a small approximation parameter. Then there exists a poly$(m,n, 1/\delta)$-time randomized algorithm which with probability at least 0.1 finds a point $x \in [-1,+1]^n$ such that
        \begin{enumerate}[(i)]
        \item $\abs{\ip{v_j}{x-x^{(0)}}} \leq c_j\norm{v_j}_2$ for all $j\in[m]$,
        \item $\abs{x_i} \geq 1-\delta$ for at least $n/2$ indices $i \in [n]$.
        \end{enumerate}
    \end{thm}

We now give the details for Theorem~\ref{thm:spherical-komlos}.
    
    \begin{proof}
        We will invoke Theorem~\ref{thm:LovettMeka}. Let $W$ be the given matrix, and let the $v_j$ be the rows of $W$. Let $B$ be the set of $v_j$ with norm at least $\norm{v_j} \geq 20$. Set the starting point $x^{(0)} = 0$ and $\delta =0.5$, and set the parameters
        \[c_j = \left\{
        \begin{array}{lr}
            0 & j \in B\\
            50 & j \not\in B
        \end{array}\right.\]
        We must check $\sum_{j} \exp(-c_j^2/16) \leq \frac{n}{32}$. Because of the unit norm constraints on the columns of $W$, the set $B$ isn't that big. The sum $\sum_j \norm{v_j}^2$ is at most $n$, and hence at most $n/400$ of the $v_j$ have $\norm{v_j} \geq 20$. Thus we can bound
        \begin{equation*}
        \begin{split}
            \sum_{j} \exp(-c_j^2/16) & = \sum_{j \in B} \exp(-c_j^2/16) + \sum_{j \not \in B} \exp(-c_j^2/16) \\
            &\leq \frac{n}{400} + n\cdot \exp(-50^2/16) 
            \leq \frac{n}{32}.
            \end{split}
        \end{equation*}

        Invoking the theorem, we can algorithmically produce $x$. We will show $\widehat x = \frac{x}{\norm{x}}$ satisfies our claim. Because at least half of the coordinates of $x$ are $+1$ or $-1$, we have $\norm{x} = \Omega(\sqrt{n})$. For rows in $B$, the inner product between $x$ and $v_j$ is exactly zero, and hence these rows are definitely $O(1/\sqrt{n})$. For rows outside of $B$, the norm of $v_j$ is $O(1)$ and from the theorem we have $\ip{x}{v_j} \leq O(1)\cdot \norm{v_j} = O(1)$. Therefore, the inner product with the normalized vector $\widehat{x}$ is $O\left(1/\sqrt{n}\right)$ as desired. 
    \end{proof}

\section{Hardness of Approximation.}
\label{sec:hardness-proof}

In this section we show that the \textsc{Spherical Discrepancy} problem is APX-hard. Here is the formal specification of the \textsc{Spherical Discrepancy} problem,
    
    \begin{tabular}{l}
        \textsc{Spherical Discrepancy}\\
        \textbf{Input:} a collection of unit vectors $v_1,v_2, \dots, v_m \in S^{n-1}$\\
        \textbf{Output:} compute the minimum value of $\max_i\innerproduct{x}{v_i}$ for $x\in S^{n-1}$
    \end{tabular}
    
    \noindent We say that an algorithm is an \textit{$\alpha$-approximation} or \textit{$\alpha$-approximates \textsc{Spherical Discrepancy}} if it outputs $x \in S^{n-1}$ which is within a multiplicative $\alpha$ factor of the best possible $x$,
    \[\max_i \innerproduct{v_i}{x} \leq \alpha \cdot \min_{x \in S^{n-1}} \max_i \innerproduct{v_i}{x} \]
    We prove a constant factor hardness of distinguishing result for \textsc{Spherical Discrepancy},
    
    \apxHardness*
    
    \begin{corollary}
        For some constant $C > 1$, it is NP-hard to $C$-approximate \textsc{Spherical Discrepancy}.
    \end{corollary}
    
    Previous work by Petkovi{\'c} et al~\cite{sphereCoveringQuadraticProgramming} using a different approach showed that solving \textsc{Spherical Discrepancy} exactly is NP-hard when the $v_i$ are not restricted to be unit vectors. 
        
    A few words are in order about the notion of approximation algorithm. For boolean and spherical discrepancy the notion of $\alpha$-approximation is slightly orthogonal to the goals of papers such as~\cite{spencer85,BanaszczykBlues}. This is because for many discrepancy problems, we don't even know the worst-case value of the optimum solution across the possible inputs,
    \[ \max_{v_1, v_2, \dots, v_m} \min_{x \in S^{n-1}} \max_i \innerproduct{x_i}{v}.\]
    
    Theorems such as Spencer's or the Koml{\'o}s conjecture are attempts to prove (or, in the case of an algorithm, algorithmically certify) an upper bound on the above quantity; however, they often provide no guarantee of $\alpha$-approximation because they do not relate the found algorithmic solution to the value of an optimum solution for the input. In some ways this is the ``interesting'' challenge for discrepancy because the problem is much harder when the algorithm is forced to make a guarantee that depends on the optimum solution: for example it is NP-hard to distinguish boolean set systems that have discrepancy zero from those with discrepancy $\Omega(\sqrt{n})$~\cite{DiscrepancyHardness}. We conjecture in Section~\ref{sec:questions} that for \textsc{Spherical Discrepancy} the situation is similar and the true approximation factor is significantly worse than what is proven here.
    
    The reduction used to prove Theorem~\ref{thm:hardness-of-apx} is a gap-preserving gadget reduction from \textsc{Max NAE-E3-SAT}. In the \textsc{Max NAE-E3-SAT} problem, which stands for Not-All-Equal Exactly-3 SAT, we are given a collection of $m$ clauses each of which involves exactly three distinct literals. An assignment to the variables satisfies a clause if the assignments to all literals are not all the same. The instances must also have the number of occurrences of each variable bounded by some absolute constant $B$. An observation made by Charikar et al~\cite[Theorem 11]{CharikarClustering} states that (even further restricting the SAT instance to be monotone) there are constants $\gamma < 1$ and $B$ so that it is NP-hard to distinguish an instance which is satisfiable, from one in which at most $\gamma m$ clauses can be simultaneously satisfied. Observe that the size of these instances is guaranteed to be linear, $m \leq \frac{Bn}{3} = O(n)$, by construction. Now we prove Theorem~\ref{thm:hardness-of-apx}.
    
    
    \begin{proof}
    Let $C_1, \dots, C_m$ be a hard instance of \textsc{Max NAE-E3-SAT}. The dimension for our \textsc{Spherical Discrepancy} instance will be $n$. Construct the instance as follows:
    \begin{itemize}
        \item For each $C_i$ take its $\{0, \pm 1\}$ indicator vector $\bone_{C_i}$ for whether a variable occurs in the clause and whether the variable is negated in the clause. Add the two vectors $\frac{1}{\sqrt{3}}\bone_{C_i}$ and  $\frac{-1}{\sqrt{3}}\bone_{C_i}$.
        \item Add the vectors $e_i$ and $-e_i$ for $i = 1,2, \dots, n$.
    \end{itemize}
    It is clear that all vectors are unit vectors and the number of vectors is $O(n)$. If the \textsc{Max NAE-E3-SAT} instance is satisfiable, let $x$ be the normalized $\{\pm 1\}$-coloring. On vectors of the first type, the NAE constraint implies $\abs{\innerproduct{x}{v_i}} =  \frac{1}{\sqrt{3n}}$, whereas on the second type $\abs{\innerproduct{x}{v_i}} = \frac{1}{\sqrt{n}}$, so the value of the new instance is $\frac{1}{\sqrt{n}}$.
    
    On the other hand, assume the \textsc{Max NAE-E3-SAT} instance is far from satisfiable, and let $x$ be a unit vector. We must show the value of $x$ is at least $\frac{C}{\sqrt{n}}$ for some constant $C > 1$. On boolean coordinates $x \in \{\frac{1}{\sqrt{n}}, \frac{-1}{\sqrt{n}}\}^n$, since there is at least one clause which is not satisfied by the corresponding assignment, the value of $x$ is exactly $\sqrt{\frac{3}{n}}$, which is larger than $\frac{1}{\sqrt{n}}$, hence the value is large here. But as $x$ moves away from the set $\{\frac{1}{\sqrt{n}}, \frac{-1}{\sqrt{n}}\}^n$, the value is forced to be large because of the vectors of the second type. We convert this intuition to a proof.
    
    Choose $C = \min\left(\sqrt{\frac{B - \frac{9}{16}(1-\gamma)}{B-(1-\gamma)}}, \frac{3\sqrt{3}}{4}\right)$. If $\norm{x}_\infty \geq \frac{C}{\sqrt{n}}$, the value of this solution will be at least $\frac{C}{\sqrt{n}}$ via the vectors of the second type, so assume that $\norm{x}_\infty < \frac{C}{\sqrt{n}}$.
    
    Let $S = \{i : x_i < \frac{0.75}{\sqrt{n}}\}$, a set we call the ``small'' coordinates of $x$, and let $\alpha = \abs{S}/n$. Let $\widetilde{x}$ denote the coloring obtained by taking the sign of each coordinate of $x$ (assign arbitrarily for zero). There are at least $1-\gamma$ fraction of clauses not satisfied by $\widetilde{x}$. We claim that for some clause $C_i$ not satisfied by $\widetilde{x}$, none of its three variables are in $S$. This will finish the proof: because the coordinates of $x$ must have the same signs in $C_i$, and they are all big, the inner product with $\frac{1}{\sqrt{3}}\bone_{C_i}$ or $\frac{-1}{\sqrt{3}}\bone_{C_i}$ must be least $3\cdot \frac{.75}{\sqrt{3n}} \geq \frac{C}{\sqrt{n}}$. Suppose then that every unsatisfied clause by $\widetilde{x}$ has a variable from $S$. Because of the bound $B$ on the maximum number of occurrences of any variable, $\alpha \geq \frac{1-\gamma}{B}$. Outside of $S$ we use the bound on $\norm{x}_\infty$ to bound the big coordinates,
    \begin{equation*}
        \begin{split}
        \sum_{i=1}^n x_i^2 &< \left(\frac{C}{\sqrt{n}}\right)^2(1-\alpha) \cdot n + \left(\frac{0.75}{\sqrt{n}}\right)^2 \alpha \cdot n \\ 
        & \leq C^2\left(1-\frac{1-\gamma}{B}\right) + \frac{9}{16}\cdot \frac{1-\gamma}{B} 
        \end{split}
    \end{equation*}
     The choice of $C \leq \sqrt{\frac{B - \frac{9}{16}(1-\gamma)}{B-(1-\gamma)}}$ shows the right-hand side is at most 1, contradicting that $\sum_i x_i^2 = 1$.
    \end{proof}
    
\section{Further Questions.}
\label{sec:questions}
The current work implements an algorithmic approach to covering problems in geometry, and there are several questions left open for future investigation.

\begin{itemize}

    \item Perhaps most generally, the idea to generalize $\{-1, +1\}^n$ to the $\sqrt{n}$-radius sphere may be interesting to consider for other combinatorial problems besides discrepancy. We would be interested in seeing concrete realizations of the following algorithmic strategy:
    \begin{enumerate}[(1)]
        \item Relax an optimization problem over $\{-1, +1\}^n$ to the $\sqrt{n}$-radius sphere.
        \item Solve the relaxed problem.
        \item Round the $\sqrt{n}$-radius sphere to $\{-1, +1\}^n$.
    \end{enumerate}
    For example, Trevisan~\cite{TrevisanEigenMaxCut} provides an interesting algorithm for MaxCut that exactly matches this framework. Indeed, spectral relaxations such as Trevisan's can be interpreted in this perspective. However, in general, unlike spectral algorithms, in our framework it is not immediate that the relaxed problem in step (2) can be perfectly solved! Spectral relaxation corresponds to objective functions of degree at most 2 which can be perfectly solved, but polynomial optimization of higher degree over the sphere is NP-hard in general (and even difficult to approximate~\cite{SpherePolynomialOpt}). Therefore we find it particularly interesting whether any problems which are naturally degree-4 or higher (or, as in the case of discrepancy, are not polynomials) can be attacked using this framework. The hypersphere also provides more structure for rounding than e.g. a Sum-of-Squares relaxation.

\item If we were able to improve the error bound in Algorithm~\ref{alg:ip-min}, we could remove the logarithmic factor in Theorem~\ref{thm:large-cap-density}. The conjectural ``right'' error bound for proving a linear lower bound can be computed from Proposition~\ref{fact:gaussian-inv} with $\delta = \frac{n}{m}$,
\begin{conjecture}
There is an efficient algorithm that improves the bound in Theorem~\ref{thm:ip-minimization} to
\[\innerproduct{v_i}{x} \leq \sqrt{\frac{2\ln \frac{m}{n}}{n}}\left(1 - \frac{\ln \ln \frac{m}{n}}{4\ln \frac{m}{n}} + O\left(\frac{1}{\ln \frac{m}{n}}\right)\right). \]
\end{conjecture}

\item However, even if we are able to remove the log factor, it is unclear if our techniques are able to prove the bona fide simplex bound $\tau_{n, \theta}$, instead of the bound $\Omega(n)$, which is weaker by a constant factor. It seems possible that the essence of Algorithm~\ref{alg:ip-min} could be extracted into a continuous-time walk (either deterministic or random) which avoids the lossy algorithmic analysis we went through.

\begin{question}
Can Algorithm~\ref{alg:ip-min} be refined to improve the constant in the lower bound and completely prove Conjecture~\ref{conj:cap-density}?
\end{question}

\item We should point out some fundamental geometric questions about $\tau_{n, \delta}$ that, to the best of the authors' knowledge, are open. The formal definition of $\tau_{n, \delta}$ is as follows. Let $T$ be a regular spherical simplex inscribed in a spherical cap of volume $\delta$. Let $C_i$ be caps of volume $\delta$ around vertex $i$ of $T$. Then
\[\tau_{n, \delta} \defeq \frac{\sum_{i=1}^n\vol(C_i \cap T)}{\vol(T)}  = \frac{n \cdot \vol(C_1 \cap T)}{\vol(T)}.\]
Rogers~\cite{RogersSimplexBound} computed that for $\tau_n = \lim_{\delta \to 0} \tau_{n, \delta}$ the Euclidean covering density, $\tau_n \sim \frac{n}{e\sqrt{e}}$. It is natural to conjecture that the densities are always linear,

\begin{conjecture}
\label{conj:density-linearity}
$\tau_{n, \delta} \geq c\cdot n$ for some positive constant $c$.
\end{conjecture}
The conjecture is verified in the ``nearly Euclidean" regime by B{\"o}r{\"o}czky and Wintsche~\cite[Example 6.3]{BoroczkyEqualCovering}.  
Intuitively, as the simplex becomes more curved, the relative volume near the center of the simplex increases. The caps $C$ contain approximately half of the local mass near the center, and so we might expect the cap to contain more and more of the simplex as it becomes more curved, up to the limit where $T$ is a hemisphere and each cap equals half of the simplex, $\tau_{n, 1/2} = n/2$. Based on this intuition we conjecture monotonicity of $\tau_{n, \delta}$,

\begin{conjecture}
\label{conj:density-monotonicity}
For every $n$, as a function of the dihedral angle $\theta \in [\arccos(1/(n-1)), \pi]$ the function $\tau_{n, \theta}$ is monotonically increasing. 
\end{conjecture}
Here we have changed notation, so that $\theta$ is the angle between two planes defining the simplex, and we have extended the conjecture to include hyperbolic space (the minimum dihedral angle in $n$-dimensional hyperbolic space is $\arccos(1/(n-1))$, achieved by the ideal regular simplex).

One must be careful, however, because the related expression for the simplicial packing density is a \textit{decreasing} function of $\theta$ (for sufficiently large $n$), as proven by Marshall~\cite{MarshallDensityMonotonicity} and Kellerhals~\cite{KellerhalsDensityMonotonicity} (though we could not verify the proof outside of hyperbolic space). The proof is analytic, and Kellerhals poses as an open problem to find a geometric proof, which we also pose as a challenge for Conjectures~\ref{conj:density-linearity} and~\ref{conj:density-monotonicity}.

\item The section on sphere covering presented Conjecture~\ref{conj:gaussian-density},

\gaussianDensityConjecture*

It may be possible to adapt the boolean discrepancy ``sample and project'' algorithm of Rothvoss~\cite{RothvossProjection} or Eldan and Singh~\cite{EldanSingh} to prove this lower bound.

\item We briefly pointed out in Theorem~\ref{thm:generate-sphere-packing} that Algorithm~\ref{alg:ip-min} can also be used to generate a set of points in $S^{n-1}$ that are relatively spaced out. Therefore we ask,
\begin{question}
Can we use Algorithm~\ref{alg:ip-min} to deterministically build smaller hitting sets than Rabani-Shpilka~\cite{RabaniShpilkaHittingSets}? For every $\delta \geq 2^{-o(\sqrt{n})}$, can we deterministically generate a sphere cover using spherical caps of volume $\delta$ with density $\widetilde{O}(n)$ in time $\poly(n, 1/\delta)$?
\end{question}
The problem with using Algorithm~\ref{alg:ip-min} in its current form is that Theorem~\ref{thm:ip-minimization} is a ``packing property''; we have no guarantee that $m$ points will cover the whole sphere.

\item Algorithm~\ref{alg:ip-min} minimizes the max of a collection of linear functions on the sphere. Can we adapt it to minimize functions that are ``slightly nonlinear'', or sets with boundaries that are slightly nonlinear? An interesting question arises if we consider sets with a diameter bound.

Fix a parameter $\theta \in [0, \pi]$, and define a distance graph in spherical space $S^{n}_{\geq \theta}$ with vertex set $S^{n}$ and edge set \[E(S^{n}_{\geq\theta}) \defeq \{ (x, y) \mid \innerproduct{x}{y} \geq \theta \}.\]
Independent sets in $S^{n}_{\geq \theta}$ are sets with diameter at most $\theta$, and therefore a cover of $S^n$ by spherical caps with diameter $\theta$ yields a finite coloring of $S^n_{\geq \theta}$. It is hopeful that this is the best possible up to a constant,


\begin{conjecture}
For every $\theta$, $\chi(S^{n}_{\geq \theta}) \geq \Omega(B_{n, \theta})$, where $B_{n, \theta}$ equals the minimum number of spheres of radius $\theta$ needed to cover $S^n$.
\end{conjecture}




\item The basic problem of \textsc{Spherical Discrepancy} has constant factor hardness of approximation as we showed in Section~\ref{sec:hardness-proof} but it seems likely that the problem has a much worse approximation factor. 
\begin{conjecture}
\label{conj:hardness-of-apx}
For every $16n \leq m \leq 2^{\sqrt{n}}$, it is NP-hard to approximate \textsc{Spherical Discrepancy} on $m$ unit vectors within a factor of $\Omega\left(\sqrt{n \ln \frac{m}{n}}\right)$.
\end{conjecture}
Evidence for this conjecture comes from the boolean regime, where despite the fact that every set system on $O(n)$ sets has discrepancy $O(\sqrt{n})$, it is NP-hard to distinguish set systems with zero discrepancy from those with discrepancy $\Omega(\sqrt{n})$~\cite{DiscrepancyHardness}. In the spherical domain, on the other hand, given a set of unit vectors it is easy to check if there is a discrepancy zero vector i.e. a vector orthogonal to the entire set. It is not clear how the \textsc{Spherical Discrepancy} problem behaves when we promise a lower bound on the solution value such as $\frac{1}{n}$ in order to avoid issues of this sort.
\end{itemize}

\section{Acknowledgments}

We thank Will Perkins for clarifying the extent to which the proofs in \cite{covering59} apply to spherical cap lower bounds, which inspired us to pursue this work further. We thank Andy Drucker for many discussions and guidance on this work, and for introducing us to discrepancy theory.

\bibliographystyle{alpha}
\bibliography{bibliography}

\newcommand{\etalchar}[1]{$^{#1}$}
\begin{thebibliography}{BGL{\etalchar{+}}17}

\bibitem[AHK12]{MWSurvey}
Sanjeev Arora, Elad Hazan, and Satyen Kale.
\newblock The multiplicative weights update method: a meta-algorithm and
  applications.
\newblock {\em Theory Comput.}, 8:121--164, 2012.

\bibitem[Ban12]{DiscrepancySDPSurvey}
Nikhil Bansal.
\newblock Semidefinite optimization in discrepancy theory.
\newblock {\em Math. Program.}, 134(1, Ser. B):5--22, 2012.

\bibitem[BDG19]{AlgorithmicWeakKomlos}
Nikhil Bansal, Daniel Dadush, and Shashwat Garg.
\newblock An algorithm for {K}oml\'{o}s conjecture matching {B}anaszczyk's
  bound.
\newblock {\em SIAM J. Comput.}, 48(2):534--553, 2019.

\bibitem[BDGL18]{BanaszczykBlues}
Nikhil Bansal, Daniel Dadush, Shashwat Garg, and Shachar Lovett.
\newblock The {G}ram-{S}chmidt walk: a cure for the {B}anaszczyk blues.
\newblock In {\em {P}roceedings of the 50th {A}nnual {S}ymposium on {T}heory of
  {C}omputing}, pages 587--597. 2018.

\bibitem[BDM12]{KissingNumbersSurvey}
Peter Boyvalenkov, Stefan Dodunekov, and Oleg Musin.
\newblock A survey on the kissing numbers.
\newblock {\em Serdica Math. J.}, 38(4):507--522, 2012.

\bibitem[BEJ76]{GaussianInvApprox}
J.~M. Blair, C.~A. Edwards, and J.~H. Johnson.
\newblock Rational {C}hebyshev approximations for the inverse of the error
  function.
\newblock {\em Math. Comp.}, 30(136):827--830, 1976.

\bibitem[BG17]{BansalGargBeyondPartialColoring}
Nikhil Bansal and Shashwat Garg.
\newblock Algorithmic discrepancy beyond partial coloring.
\newblock In {\em {P}roceedings of the 49th {A}nnual {S}ymposium on {T}heory of
  {C}omputing}, pages 914--926. 2017.

\bibitem[BGL{\etalchar{+}}17]{SpherePolynomialOpt}
Vijay Bhattiprolu, Venkatesan Guruswami, Euiwoong Lee, Mrinalkanti Ghosh, and
  Madhur Tulsiani.
\newblock Weak decoupling, polynomial folds, and approximate optimization over
  the sphere.
\newblock In {\em 58th {A}nnual {IEEE} {S}ymposium on {F}oundations of
  {C}omputer {S}cience---{FOCS} 2017}, pages 1008--1019. IEEE Computer Soc.,
  Los Alamitos, CA, 2017.

\bibitem[B{\"o}r78]{BoroczkyPacking}
K\'{a}roly B{\"o}r\"{o}czky, Jr.
\newblock Packing of spheres in spaces of constant curvature.
\newblock {\em Acta Math. Acad. Sci. Hungar.}, 32(3-4):243--261, 1978.

\bibitem[B{\"o}r04]{BoroczkyFinitePackingAndCovering}
K\'{a}roly B{\"o}r\"{o}czky, Jr.
\newblock Finite packing and covering.
\newblock 154:xviii+380, 2004.

\bibitem[BW03]{BoroczkyEqualCovering}
K\'{a}roly B\"{o}r\"{o}czky, Jr. and Gergely Wintsche.
\newblock Covering the sphere by equal spherical balls.
\newblock In {\em Discrete and Computational Geometry: the Goodman-Pollack
  Festschrift}, pages 235--251. Springer, 2003.

\bibitem[CFR59]{covering59}
H.~S.~M. Coxeter, L.~Few, and C.~A. Rogers.
\newblock Covering space with equal spheres.
\newblock {\em Mathematika}, 6:147--157, 1959.

\bibitem[CGW05]{CharikarClustering}
Moses Charikar, Venkatesan Guruswami, and Anthony Wirth.
\newblock Clustering with qualitative information.
\newblock {\em J. Comput. System Sci.}, 71(3):360--383, 2005.

\bibitem[Cha00]{ChazelleDiscrepancyBook}
Bernard Chazelle.
\newblock {\em The discrepancy method: randomness and complexity}.
\newblock Cambridge University Press, Cambridge, 2000.

\bibitem[CL09]{3dCoveringDataStructure}
Fr{\'e}d{\'e}ric Cazals and S{\'e}bastien Loriot.
\newblock Computing the arrangement of circles on a sphere, with applications
  in structural biology.
\newblock {\em Computational Geometry}, 42(6-7):551--565, 2009.

\bibitem[CNN11]{DiscrepancyHardness}
Moses Charikar, Alantha Newman, and Aleksandar Nikolov.
\newblock Tight hardness results for minimizing discrepancy.
\newblock In {\em Proceedings of the Twenty-second Annual Symposium on Discrete
  Algorithms}, pages 1607--1614, 2011.

\bibitem[Cox63]{CoxeterKissing}
H.~S.~M. Coxeter.
\newblock An upper bound for the number of equal nonoverlapping spheres that
  can touch another of the same size.
\newblock In {\em Proc. {S}ympos. {P}ure {M}ath., {V}ol. {VII}}, pages 53--71.
  Amer. Math. Soc., Providence, R.I., 1963.

\bibitem[DF87]{SphereCoordinateHistory}
Persi Diaconis and David Freedman.
\newblock A dozen de {F}inetti-style results in search of a theory.
\newblock {\em Ann. Inst. H. Poincar\'{e} Probab. Statist.}, 23(2,
  suppl.):397--423, 1987.

\bibitem[ES18]{EldanSingh}
Ronen Eldan and Mohit Singh.
\newblock Efficient algorithms for discrepancy minimization in convex sets.
\newblock {\em Random Structures Algorithms}, 53(2):289--307, 2018.

\bibitem[JJP18]{PerkinsKissingNumbers}
Matthew Jenssen, Felix Joos, and Will Perkins.
\newblock On kissing numbers and spherical codes in high dimensions.
\newblock {\em Adv. Math.}, 335:307--321, 2018.

\bibitem[Kel98]{KellerhalsDensityMonotonicity}
Ruth Kellerhals.
\newblock Ball packings in spaces of constant curvature and the simplicial
  density function.
\newblock {\em J. Reine Angew. Math.}, 494:189--203, 1998.
\newblock Dedicated to Martin Kneser on the occasion of his 70th birthday.

\bibitem[Li11]{capVolume}
S.~Li.
\newblock Concise formulas for the area and volume of a hyperspherical cap.
\newblock {\em Asian J. Math. Stat.}, 4(1):66--70, 2011.

\bibitem[LM15]{LovettMeka}
Shachar Lovett and Raghu Meka.
\newblock Constructive discrepancy minimization by walking on the edges.
\newblock {\em SIAM J. Comput.}, 44(5):1573--1582, 2015.

\bibitem[Lov00]{LovaszVecDisc}
L\'{a}szl\'{o} Lov\'{a}sz.
\newblock Integer sequences and semidefinite programming.
\newblock {\em Publ. Math. Debrecen}, 56(3-4):475--479, 2000.
\newblock Dedicated to Professor K\'{a}lm\'{a}n Gy\H{o}ry on the occasion of
  his 60th birthday.

\bibitem[LRR17]{RothvossMW}
Avi Levy, Harishchandra Ramadas, and Thomas Rothvoss.
\newblock Deterministic discrepancy minimization via the multiplicative weight
  update method.
\newblock In {\em Integer programming and combinatorial optimization}, pages
  380--391. 2017.

\bibitem[Mar99]{MarshallDensityMonotonicity}
T.~H. Marshall.
\newblock Asymptotic volume formulae and hyperbolic ball packing.
\newblock {\em Ann. Acad. Sci. Fenn. Math.}, 24(1):31--43, 1999.

\bibitem[Nik13]{VectorKomlos}
Aleksandar Nikolov.
\newblock The {K}oml{\'o}s conjecture holds for vector colorings.
\newblock E-print, arXiv:1301.4039, 2013.

\bibitem[PPL12]{sphereCoveringQuadraticProgramming}
Marko~D. Petkovi\'{c}, Dragoljub Pokrajac, and Longin~Jan Latecki.
\newblock Spherical coverage verification.
\newblock {\em Appl. Math. Comput.}, 218(19):9699--9715, 2012.

\bibitem[Rog58]{RogersSimplexBound}
C.~A. Rogers.
\newblock The packing of equal spheres.
\newblock {\em Proc. London Math. Soc. (3)}, 8:609--620, 1958.

\bibitem[Rot17]{RothvossProjection}
Thomas Rothvoss.
\newblock Constructive discrepancy minimization for convex sets.
\newblock {\em SIAM J. Comput.}, 46(1):224--234, 2017.

\bibitem[RS09]{RabaniShpilkaHittingSets}
Yuval Rabani and Amir Shpilka.
\newblock Explicit construction of a small epsilon-net for linear threshold
  functions.
\newblock In {\em {P}roceedings of the 41st Annual {S}ymposium on {T}heory of
  {C}omputing}, pages 649--658. 2009.

\bibitem[Spe85]{spencer85}
Joel Spencer.
\newblock Six standard deviations suffice.
\newblock {\em Trans. Amer. Math. Soc.}, 289(2):679--706, 1985.

\bibitem[Spr07]{SpruillSphereUnif}
M.~C. Spruill.
\newblock Asymptotic distribution of coordinates on high dimensional spheres.
\newblock {\em Electron. Comm. Probab.}, 12:234--247, 2007.

\bibitem[Sta82]{StamSphereUnif}
A.~J. Stam.
\newblock Limit theorems for uniform distributions on spheres in
  high-dimensional {E}uclidean spaces.
\newblock {\em J. Appl. Probab.}, 19(1):221--228, 1982.

\bibitem[SV05]{NonoptimalityOfPackingToCovering}
Achill Sch\"{u}rmann and Frank Vallentin.
\newblock Local covering optimality of lattices: {L}eech lattice versus root
  lattice {$E_8$}.
\newblock {\em Int. Math. Res. Not.}, (32):1937--1955, 2005.

\bibitem[Tre12]{TrevisanEigenMaxCut}
Luca Trevisan.
\newblock Max cut and the smallest eigenvalue.
\newblock {\em SIAM J. Comput.}, 41(6):1769--1786, 2012.

\end{thebibliography}

\appendix
\section{\null}

\expIneq*
\begin{proof}
\begin{equation*}
    \begin{split}
        e^x & = 1 + x + \frac{x^2}{2} + \sum_{k=3}^\infty \frac{x^k}{k!}\\
        & \leq 1 + x + \frac{x^2}{2} + x^3\left(\sum_{k=3}^\infty\frac{1}{k!}\right)\\
        & = 1 + x + \frac{x^2}{2} + x^3\left(e - 2.5\right)\\
        & \leq 1 + x + \frac{x^2}{2} + \frac{x^3}{2}
    \end{split}
\end{equation*}
\end{proof}

\gaussianTail*
\begin{proof}
\[\overline{\Phi}(t) = \frac{1}{2\pi} \int_t^\infty e^{-x^2/2}\, dx = \varphi(t) \int_0^\infty e^{-x^2/2}\cdot e^{-xt}\, dx \]
For the upper bound,
\[ \overline{\Phi}(t)\leq \varphi(t)\int_0^\infty e^{-xt}\, dx = \frac{\varphi(t)}{t}. \]
For the lower bound, since $\varphi$ decreases so quickly the integral is well-approximated by just the first unit interval,
\[ \overline{\Phi}(t)\geq \varphi(t)\int_0^1 e^{-1/2}e^{-xt}\, dx = \frac{\varphi(t)}{t} \cdot e^{-1/2}(1 - e^{-t}). \]
Assuming that $t \geq 1$ ensures $C = e^{-1/2}(1 - 1/e)$ works.
\end{proof}

\begin{restatable}{proposition}{expFidelity}
\label{fact:exp-fidelity}
For $x = o(\sqrt{n})$,
\[\left(1 + \frac{x}{n}\right)^n \sim e^x \]
For $x = \Theta(\sqrt{n})$,
\[\left(1 + \frac{x}{n}\right)^n = \Theta(e^x) \]
For $x = \omega(\sqrt{n})$,
\[\left(1 + \frac{x}{n}\right)^n = \omega(e^x) \]
\end{restatable}
\begin{proof}
\begin{align*}
    \left(1 + \frac{x}{n}\right)^n &= \exp\left(\ln\left(\left(1 + \frac{x}{n}\right)^n\right)\right)\\
    & = \exp\left(n \cdot\ln\left(1+\frac{x}{n}\right)\right)\\
    & = \exp\left( n \cdot \left(\frac{x}{n} + \Theta\left(\frac{x^2}{n^2}\right)\right)\right)\\
    & = e^x \cdot \exp(\Theta(x^2/n))
\end{align*}
\end{proof}

\begin{restatable}{proposition}{gaussianInv}
\label{fact:gaussian-inv}
As $\delta \to 0$,
\begin{align*}
\overline{\Phi}^{-1}(\delta) &= \sqrt{2\ln 1/\delta}- \frac{\ln 2\sqrt{\pi}}{\sqrt{2\ln 1/\delta}}\\
 & \quad - \frac{\ln\ln 1/\delta}{2\sqrt{2\ln 1 /\delta}} + o\left(\frac{1}{\sqrt{\ln 1/\delta}}\right)
\end{align*}
\end{restatable}
\begin{proof}
The complementary error function $\erfc$ is related to $\overline{\Phi}$ via $\erfc(z) = 2\overline{\Phi}(\sqrt{2}z)$. Blair et al~\cite{GaussianInvApprox} show an asymptotic formula for the inverse complementary error function,
\begin{align*}
    \erfc^{-1}(\delta)^2 &= 
    \ln 1/\delta  - \ln \sqrt{\pi} - \frac{1}{2}\ln(\ln 1/\delta - \ln \sqrt{\pi}) \\ 
    & \quad + o(1)\\
    \overline{\Phi}^{-1}(\delta)^2
    &= 2\ln 1/\delta  - 2\ln 2\sqrt{\pi} - \ln(\ln 1/\delta - \ln 2\sqrt{\pi}) \\
    & \quad + o(1)\\
    &= 2\ln 1/\delta  - 2\ln 2\sqrt{\pi} - \ln\ln 1/\delta + o(1).
\end{align*}
Taking a square root,
\begin{align*}
    &= \sqrt{2\ln 1/\delta}\left(1 - \frac{2\ln 2\sqrt{\pi} + \ln\ln 1/\delta+ o(1)}{2\ln 1 /\delta}\right)^{1/2} \\
    & = \sqrt{2\ln 1/\delta}\left(1 - \frac{2\ln2 \sqrt{\pi} + \ln\ln 1/\delta}{4\ln 1 /\delta} + o\left(\frac{1}{\ln 1/\delta}\right)\right)\\
    &= \sqrt{2\ln 1/\delta}- \frac{\ln 2\sqrt{\pi}}{\sqrt{2\ln 1/\delta}} - \frac{\ln\ln 1/\delta}{2\sqrt{2\ln 1 /\delta}} \\
    & \quad + o\left(\frac{1}{\sqrt{\ln 1/\delta}}\right)
\end{align*}
\end{proof}

\end{document}